\documentclass[11pt]{article}
\usepackage{amsmath,amsthm,amssymb,verbatim,color}
\usepackage{graphicx}
\usepackage{wrapfig}
\usepackage{subfigure}
\newtheorem{theorem}{Theorem}
\newtheorem{lemma}[theorem]{Lemma}
\newtheorem{fact}[theorem]{Fact}

\newtheorem{conjecture}{Conjecture}
\newtheorem{definition}{Definition}
\newtheorem{construction}{Construction}
\newtheorem*{acknowledgement}{Acknowledgement}

\def\beq{\begin{equation}}\def\eeq{\end{equation}}
\def\beqn{\begin{eqnarray}}\def\eeqn{\end{eqnarray}}

\def\({\mbox{$($}}\def\){\mbox{$)$}}

\def\qed{\ifhmode\unskip\nobreak\fi\quad\ifmmode\Box\else$\Box$\fi}
\begin{document}
\title{Perfect Matchings in 4-uniform hypergraphs}
\author{Imdadullah Khan\\ 
\small Department of Computer Science \\[-0.8ex]
\small Rutgers University\\[-0.8ex]
\small New Brunswick, NJ, USA 08903\\[-0.8ex]
\small \texttt{imdadk@cs.rutgers.edu}\\[-0.8ex]
}
\date{January 10, 2011}

\maketitle

\begin{abstract}
A perfect matching in a $4$-uniform hypergraph is a subset of $\lfloor\frac{n}{4}\rfloor$ disjoint edges. We prove that if $H$ is a sufficiently large $4$-uniform hypergraph on $n=4k$ vertices such that every vertex belongs to more than ${n-1\choose 3} - {3n/4 \choose 3}$ edges then $H$ contains a perfect matching. This bound is tight and settles a conjecture of H{\'a}n, Person and Schacht.
\end{abstract}

\section{Introduction and Notation}
For graphs we follow the notation in \cite{B1}. For a set $T$, we refer to all of its $k$-element subsets ($k$-sets for short) as ${T \choose k}$ and the number of such $k$-sets as ${|T| \choose k}$. We say that $H = (V(H),E(H))$ is an $r$-uniform hypergraph or $r$-graph for short, where $V(H)$ is the set of vertices and $E\subset {V(H) \choose r}$ is a family of $r$-sets of $V(H)$. When the graph referred to is clear from the context we will use $V$ instead of $V(H)$ and will identify $H$ with $E(H)$. For an $r$-graph $H$ and a set $D = \{v_1,\ldots,v_d\} \in {V \choose d}, 1\leq d \leq r$, the degree of $D$ in $H$, $deg_H(D) = deg_r(D)$ denotes the number of edges of $H$ that contain $D$. For $1\leq d \leq r$, let $$\delta_d =\delta_d(H)= \min\left\{deg_r(D) \; : \; D\in {V \choose d}\right\}$$ 

\noindent We say that $H(V_1,\ldots,V_r)$ is an $r$-partite $r$-graph, if there is a partition of $V(H)$ into $r$ sets, i.e. $V(H) = V_1\cup \cdots\cup V_r$ and every edge of $H$ uses exactly one vertex from each $V_i$. We call it a balanced $r$-partite $r$-graph if all $V_i$'s are of the same size. We say $H(V_1,\ldots,V_r)$ is a complete $r$-partite $r$-graph if every $r$-tuple that uses one vertex from each $V_i$ belongs to $E(H)$. We denote a complete balanced $r$-partite $r$-graph by $K^{(r)}(t)$, where $t = |V_i|$. For $r=3$, we refer to the balanced $3$-partite $3$-graph $H(V_1,V_2,V_3)$, where $|V_i|=4$ as a $4\times4\times4$ $3$-graph.
\vskip4pt
\noindent For an $r$-graph $H$, when $A$ and $B$ make a partition of $V(H)$, for a vertex $v\in A$ we denote by $deg_r\left(v,{B\choose r-1}\right)$ the number of $(r-1)$-sets of $B$ that make edges with $v$ while $e_r\left(A,{B\choose r-1}\right)$ is the sum of $deg_r\left(v,{B\choose r-1}\right)$ over all $v\in A$ and $d_r\left(A,{B\choose r-1}\right) =e_r\left(A,{B\choose r-1}\right)/|A|{|B|\choose r-1}$. We denote by $H\left(A,{B\choose r-1}\right)$, such an $r$-graph, when all edges use one vertex from $A$ and $r-1$ vertices from $B$. Similarly $H\left(A, B, {C\choose r-2}\right)$ is an $r$-graph where $A$, $B$ and $C$ make a partition of $V(H)$, and every edge in $H$ uses one vertex each from $A$ and $B$ and $r-2$ vertices from $C$. Degrees of vertices in $A$ and $B$ are similarly defined as above. The density of $H\left(A, B, {C\choose r-2}\right)$ is $$d_r\left(A, B,{C\choose r-2}\right) = \frac{\left|E\left(H\left(A, B,{C\choose r-2}\right)\right)\right|}{|A||B|{|C|\choose r-2}}$$ 

\noindent When $A_1,\ldots,A_r$ make a partition of $V(H)$, for a vertex $v\in A_1$ we denote by $deg_r(v, (A_2\times\cdots\times A_r))$ the number of edges in the $r$-partite $r$-graph induced by subsets $\{v\},A_2,\ldots, A_r$, and $e(A_1, (A_2\times\cdots\times A_r))$ is the sum of $deg_r(v,(A_2\times\cdots\times A_r))$ over all $v \in A_1$. Similarly $$d_r(A_1, (A_2\times\cdots\times A_r)) = \frac{e(A_1, (A_2\times\cdots\times A_r))}{|A_1\times A_2\times\cdots\times A_r|}$$ 

An $r$-graph $H$ on $n$ vertices is $\eta$-{\em dense} if it has at least $\eta {n \choose r}$ edges. We use the notation $d_r(H) \geq \eta$ to refer to an $\eta$-{\em dense} $r$-graph $H$. For $U\subset V$, $H|_U$ is the restriction of $H$ to $U$.  For simplicity we refer to $d_r(H|_U)$ as $d_r(U)$ and to $E(H|_U)$ as $E(U)$. A matching in $H$ is a set of disjoint edges of $H$ and a perfect matching is a matching that contains all vertices. Moreover we will only deal with $r$-graphs on $n$ vertices where $n=rk$ for some integer $k$, we denote this by $n\in r\mathbb{Z}$. Throughout the paper $\log$ denotes the base 2 logarithm.

\begin{definition}
 Let $d,r$ and $n$ be integers such that $1\leq d < r$, and $n\in r\mathbb{Z}$. Denote by $m_d(r,n)$ the smallest integer $m$, such that every $r$-graph $H$ on $n$ vertices with $\delta_d(H) \geq m$ contains a perfect matching. 
\end{definition}

For graphs ($r=2$), by the Dirac theorem on Hamiltonicity of graphs, it s easy to see that $m_1(2,n) \leq n/2$, and since the complete bipartite $K_{n/2-1,n/2+1}$ does not have a perfect matching we get $m_1(2,n) = n/2$.  For $r\geq 3$ and $d = r-1$, it follows from a result of R\"{o}dl, Ruci\'{n}ski and Szemer\'{e}di on Hamiltonicity of $r$-graph \cite{RRSz_HAM_ku_colDeg_approx} that $m_{r-1}(r,n) \leq n/2 + o(n)$. K\"{u}hn and Osthus \cite{KO_PM_ku_colDeg} improved this result to $m_{r-1}(r,n) \leq n/2 + 3r^2\sqrt{n\log n}$. This bound was further sharpened in \cite{RRSz_PM_ku_colDeg} to $m_{r-1}(r,n) \leq n/2 + C\log n$. In \cite{RRSz_PM_ku_colDeg_approx_better} the bound was improved to almost the true value, it was proved that $m_{r-1}(r,n) \leq n/2 + r/4$. Finally \cite{RRSz_PM_ku_colDeg_tight} settled the problem for $d=r-1$. 
\vskip 10pt
\noindent
The case $d<r-1$ is rather hard, in \cite{Pikh_PM_ku_dDeg} it is proved that for all $d\geq r/2$, $m_d(r,n)$ is close to $\frac{1}{2}{n-d \choose r-d}$. For $1\leq d<r/2$  in \cite{HPS_PM_3u_vertDeg} it was proved that $$m_d(r,n) \leq \left(\frac{r-d}{r} +o(1)\right){n-d \choose r-d}$$
A recent survey of these and other related results appear in \cite{Rod_Ruc_Survey}. In \cite{HPS_PM_3u_vertDeg} the authors posed the following conjecture. 
\begin{conjecture}[\cite{HPS_PM_3u_vertDeg}]\label{mainConjecture}
For all $1\leq d < r/2$, 
$$m_d(r,n) \sim \max\left\{\frac{1}{2}, 1-\left(\frac{r-1}{r}\right)^{r-d}\right\}{n-d \choose r-d}$$

\end{conjecture}

\noindent Note that for $r=4$ and $d=1$ the above bound yields $$m_1(3,n) \sim \frac{37}{64}{n-1 \choose 3}$$
\noindent For $r=3$ and $d=1$ this conjecture was proved in \cite{Khan_3u_vertDeg} and \cite{KOT_parallel}. Markst{\"o}m and Ruci\'{n}ski in \cite{Ruc_Marks_approx4u} improved the bound on $m_d(r,n)$ for $1\leq d < r/2$ slightly, by proving $$m_d(r,n) \leq \left(\frac{r-d}{r} -\frac{1}{r^{r-d}} + o(1)\right){n-d \choose r-d}$$ Furthermore for $r=4$ and $d=1$ in \cite{Ruc_Marks_approx4u} the authors proved that $m_1(4,n) \leq \left( \frac{42}{64} + o(1)\right){n-1\choose 3}$. In this paper we settle Conjecture \ref{mainConjecture} for the case $r=4$ and $d=1$. The main result in this paper is the following theorem. 
\begin{theorem}\label{ourMainThm}
There exist an integer $n_0$ such that if $H$ is a $4$-graph on $n \geq n_0$ ($n\in 4\mathbb{Z}$) vertices, and \beq\label{minDegree}\delta_1(H) \geq {n-1\choose 3} - {3n/4 \choose 3} + 1 \eeq  then $H$ has a perfect matching.
\end{theorem}

\noindent On the other hand the following construction from \cite{HPS_PM_3u_vertDeg} shows that this result is tight. 

\begin{construction}\label{extConstruction}
Let $A$ and $B$ be disjoint sets with $|A| = \frac{n}{4}-1$ and $|B|= n-|A|$. Let $H = (V(H),E(H))$ be a $4$-graph such that $V(H) = A\cup B$ and $E(H)$ is the set of all $4$-tuples of vertices, $T$, such that $|T\cap A|\geq 1$.
\end{construction}

\noindent We have $\delta_1(H) = {n-1\choose 3} - {3n/4\choose 3}$ (the degree of a vertex in $B$) but since every edge in a matching must use at least one vertex from $A$, the maximum matching in this graph is of size $|A|=\frac{n}{4}-1$
\begin{figure}[h!] 
\centering
\includegraphics[scale=0.55]{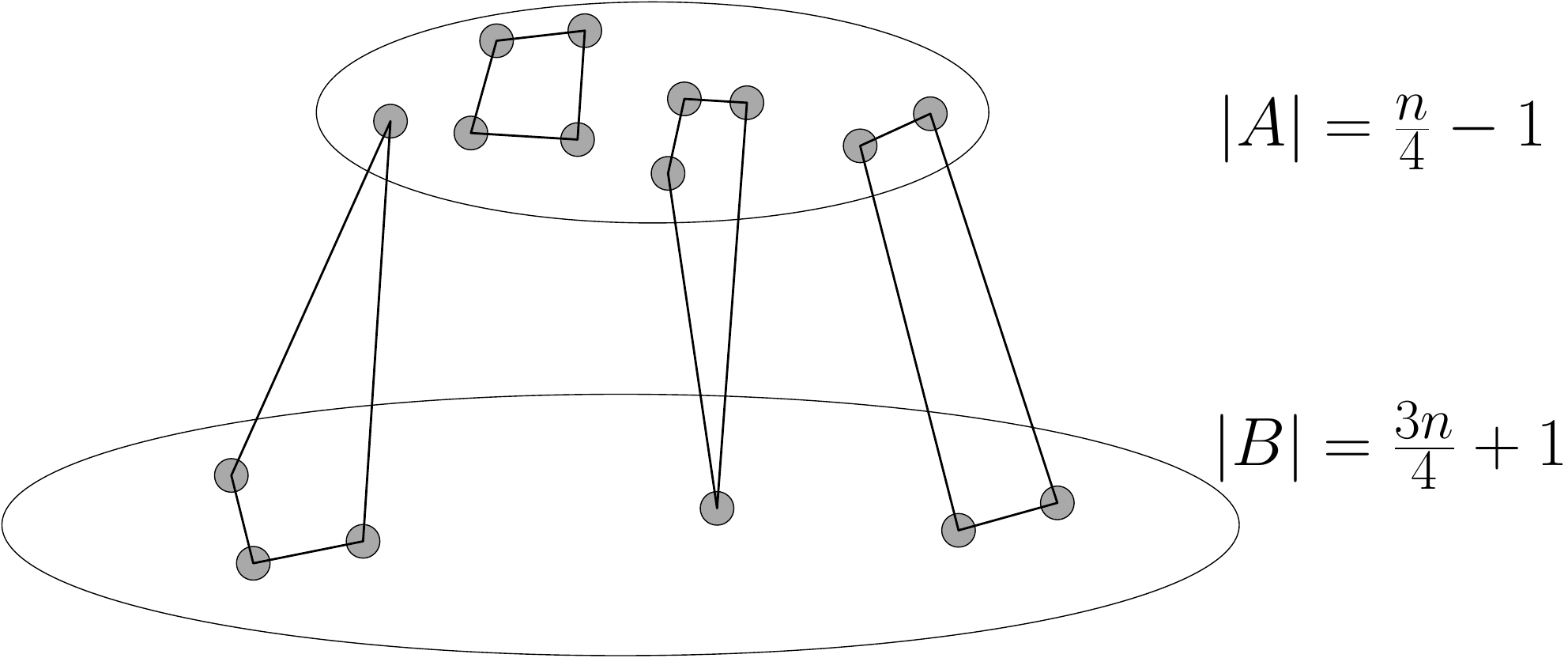} 
\caption{\footnotesize{The extremal example: A Quadrilateral represent an edge. Every edge intersects the set $A$.}}
\label{ext_example}
\end{figure}

\subsection{Outline of the proof}
We distinguish two cases to prove Theorem \ref{ourMainThm}. When $H$ is {\em `far from'} the extremal example as in Construction \ref{extConstruction} we use the {\em absorbing lemma} to find a perfect matching in $H$. The absorbing lemma (Lemma \ref{absorbLemma}) roughly states, that any $3$-graph $H$ satisfying (\ref{minDegree}), contains a small matching $M$ with the property that every ``not too large'' subset of vertices can be absorbed into $M$. In Section \ref{non_ext_case}, we first remove an absorbing matching $M$ from $H$. Then using the tools developed in section \ref{tools} and the auxiliary results of section \ref{aux_results}, we find an {\em almost perfect matching} in $H|_{V\setminus V(M)}$. The left over vertices are absorbed into $M$ to get a perfect matching in $H$. We will show that we can either find an almost perfect matching or $H$ is in the extremal case.
\vskip 6pt
For a constant $0<\alpha <1$, we say that $H$ is {\em $\alpha$-extremal}, if the following is satisfied, otherwise it is {\em $\alpha$-non-extremal}. 
\begin{definition}[Extremal Case with parameter $\alpha$]
There exists a $B\subset V(H)$ such that 
\begin{itemize}
\item $|B|\geq \left(\frac{3}{4}-\alpha\right) n$
\item $d_4\left(B\right) < \alpha$.
\end{itemize}
\end{definition}

\noindent In section \ref{extCase}, using a K{\"o}nig-Hall type argument, we build a perfect matching in $H$, when it is $\alpha$-extremal.

\section{Main tools}\label{tools}


We use the following two easy observations that will, nevertheless, be useful later on.
\vskip4pt
\begin{fact}\label{denseSubgraph}
If $G(A,B)$ is a $2\eta$-dense bipartite graph, then there must be at least $\eta|B|$ vertices in $B$ for which the degree in $A$ is at least $\eta|A|$.
\end{fact}
\vskip2pt
\noindent Indeed, otherwise the total number of edges would be less than
$$\eta |A||B| + \eta |A||B| = 2\eta |A||B|$$ a contradiction to the fact that $G(A,B)$ is $2\eta$-dense.
\vskip8pt
\begin{fact}\label{subsetsPHP}
If $G(A,B)$ is a bipartite graph, $|A| = c_1\log n$, $|B|= c_2n$ and for every vertex $ b\in B$ $deg(b,A)\geq \eta |A|$, then we can find a complete bipartite subgraph $G'(A',B')$ of $G$ such that $A'\subset A, B'\subset B, |A'|\geq \eta |A| \text{ and } |B'|\geq c_2n^{(1-c_1)}$.
\end{fact}
\vskip2pt
\noindent To see this consider the neighborhoods in $A$, of the vertices in $B$. Since there can be at most $2^{|A|} = n^{c_1}$ such
neighborhoods, by averaging there must be a neighborhood that appears for at least $\frac{c_2n}{n^{c_1}}=c_2n^{(1-c_1)}$
vertices of $B$. This means that we can find the desired complete bipartite graph.
\vskip8pt
\noindent The main tool in this paper is the following result of Erd\"os \cite{ErdosKST}, on complete balanced $r$-partite subhypergraphs of $r$-graphs.

\begin{lemma}\label{hyperKST}
For every integer $l\geq 1$ there is an integer $n_0 = n_0(r,l)$ such that: Every $r$-graph on $n> n_0$ vertices, that has at least $n^{r-1/l^{r-1}}$ edges, contains a $K^{(r)}(l)$.
\end{lemma}
\noindent In particular Lemma \ref{hyperKST} implies that for $\eta>0$ and a sufficiently small constant $\beta$, if $H$ is an $r$-graph on $n>n_0(\eta,\beta)$ vertices with $$|E(H)|\geq \eta{n\choose r}$$ then $H$ contains a $K^{(r)}(t)$, where $$t=\beta (\log n)^{1/r-1}.$$
\vskip5pt
\noindent This is so because $\eta{n\choose r} \geq \frac{n^r}{n^{1/\beta^{r-1} \log n}} = \frac{n^r}{2^{1/\beta^{r-1}}}$, when $\eta > \frac{2r!}{2^{1/\beta^{r-1}}}$.\\

\noindent The following three lemmas are repeatedly used in section \ref{non_ext_case}.
\vskip4pt

\begin{lemma}\label{subsetPHP_hyper}
Let $H(A,{B\choose 3})$ be a $4$-graph, $|A|= c_1 m$, $|B|=c_2 2^{m^3}$, for some constants $0<c_1,c_2<1$, if $$d_4\left(A,{B\choose 3}\right)\geq 2\eta$$ then there exists a complete $4$-partite $4$-graph $H'(A_1,B_1,B_2,B_3)$, with $A_1\subset A$ and $B_1,B_2\mbox{ and } B_3$ are disjoint subsets of $B$ such that $|A_1|=|B_1|=|B_2|=|B_3| = \eta|A|$.  
\end{lemma}
\vskip4pt

\begin{proof}
First apply Fact \ref{denseSubgraph} to get a subset of $3$-sets of vertices in $B$, $T_1\subset {B\choose 3}$, such that every $3$-set in $T_1$ makes edges with at least $\eta|A|$ vertices in $A$ and $|T_1|\geq \eta {|B|\choose 3} \geq \eta |B|^3/10$. Next we find a $T_2\subset T_1$, such that all $3$-sets in $T_2$ make edges with the same subset of $A$ (say $A_1\subset A$). By Fact \ref{subsetsPHP} we have $$|T_2|\geq \frac{\eta|B|^3/10}{2^{c_1m}} = \frac{\eta}{10} \frac{\left(c_22^{m^3}\right)^3}{2^{c_1m}} = \frac{c_2^3\eta}{10} 2^{m^3\left(3-c_1/m^2\right)} = \frac{c_2^3\eta}{10} \left(\frac{|B|}{c_2}\right)^{3-c_1/m^2} \geq |B|^{3-1/(\eta|A|)^2}$$ where the last inequality follows when $m$ is sufficiently large and $\eta$, $c_1$ and $c_2$ are small constants.\\

\noindent Now construct an auxiliary $3$-graph $H_1$ where $V(H_1)=B$ and edges of $H_1$ corresponds to $3$-sets in $T_2$. Applying lemma \ref{hyperKST} on $H_1$ (for $r=3$) we get a complete $3$-partite $3$-graph with color classes $B_1$, $B_2$ and $B_3$ each of size $\eta|A|$. Clearly $A_1$, $B_1$, $B_2$ and $B_3$ corresponds to color classes of a complete $4$-partite $4$-graph in $H$ as in the statement of the lemma. \hfill{} 
\end{proof}
\vskip4pt

\begin{lemma}\label{subsetPHP_hyper_2}
Let $H(A,B,{Z\choose 2})$ be a $4$-graph with $|A|=|B|=c_1 m$, $|Z|=c_2 2^{m^3}$ for some constants $0<c_1,c_2<1$, if $$d_4\left(A, B, {Z\choose 2}\right)\geq 2\eta$$ then there exists a complete $4$-partite $4$-graph $H'(A',B',Z_1,Z_2)$, with $A'\subset A$, $B'\subset B$ and  $Z_1 \mbox{ and } Z_2$ are disjoint subsets of $Z$ such that $|A'|=|B'|=|Z_1|=|Z_2|= c\log m$ for a constant $c=c(c_1,c_2,\eta)$. 
\end{lemma}
\begin{proof}
First apply Fact \ref{denseSubgraph} to get a subset of pairs of vertices in $Z$, $P_1\subset {Z\choose 2}$, such that every pair in $P_1$ makes edges with at least $\eta|A||B|$ pairs of vertices in $A\times B$ and $|P_1|\geq \eta {|Z|\choose 2} \geq \eta |Z|^2/3$. Next we find a $P_2\subset P_1$, such that all pairs in $P_2$ make edges with the same subset of pairs in $A\times B$ (say $Q_1\subset A\times B$). By Fact \ref{subsetsPHP} we have $$|P_2|\geq \frac{\eta|Z|^2/3}{2^{(c_1m)^2}} = \frac{\eta}{3} \frac{\left(c_22^{m^3}\right)^2}{2^{\left(c_1m\right)^2}} = \frac{c_2^2\eta}{3} 2^{m^3\left(2-c_1^2/m\right)} = \frac{c_2^2\eta}{3} \left(\frac{|Z|}{c_2}\right)^{2-c_1^2/m} \geq |Z|^{2-1/\left(\log m\right)}$$ where the last inequality follows when $m$ is sufficiently large, and $\eta$, $c_1$ and $c_2$ are small constants.
\vskip6pt

Now construct an auxiliary $2$-graph $G_1$ where $V(G_1)=Z$ and edges of $G_1$ corresponds to pairs in $P_2$. Applying lemma \ref{hyperKST} on $G_1$ (for $r=2$) we get a complete bipartite graph with color classes $Z_1$ and $Z_2$ each of size $\log m$.
\vskip6pt
Similarly construct an auxiliary bipartite graph $G_2$ with color classes $A$ and $B$, and edges of $G_2$ corresponds to pairs in $Q_1$. Since we have $|Q_1|\geq \eta|A||B|$ applying lemma \ref{hyperKST} on $G_2$ (for $r=2$) we get a complete bipartite graph with color classes $A'$ and $B'$, each of size $c\log m$. Clearly $A'$, $B'$ and a subset of vertices each from $Z_1$ and $Z_2$ of size $|A'|$ corresponds to color classes of the balanced complete $4$-partite $4$-graph in $H$ as in the statement of the lemma.\hfill{} 
\end{proof}
\vskip 3pt 
\begin{lemma}\label{4partVolArg}
Let $H(A,B,C,Z)$ be a $4$-partite $4$-graph with $|A| = |B| = |C| =c_1m$ and $|Z| = c_2 2^{m^3}$ for some constants $0<c_1,c_2<1$. If $$d_4(Z,(A\times B\times C)) \geq 2\eta$$ then there exists a complete $4$-partite $4$-graph $H'(A',B',C',Z')$ such that $|A'|=|B'| =|C'| =|Z'|= \beta \sqrt{\log |A|}$. 
\end{lemma}

\vskip4pt

\begin{proof}
 First apply Fact \ref{denseSubgraph} on $H$ to get a subset $Z_1$ of $Z$ such that for every vertex $z\in Z_1$, $deg_3(z,(A\times B\times C))\geq \eta |A|^3$ and $|Z_1|\geq \eta |Z|$. Now consider the auxiliary bipartite graph $G(D,Z_1)$, where $D =A\times B\times C$ and a vertex $z\in Z_1$ is connected to a $3$-set $(a,b,c)\in D$ if $(a,b,c,z)$ makes an edge of $H$. An application of Fact \ref{subsetsPHP} on $G$ gives a complete bipartite graph $G_2(D',Z'_1)$ where $$|D'|\geq \eta |A|^3 \;\;\;\mbox{    and    }\;\;\; Z'_1 \geq \eta c_2 2^{m^3(1-c_1^3)} > |A|$$ when $m$ is sufficiently large. 
\vskip6pt
\noindent Let $G_3$ be a $3$-partite $3$-graph where the color classes are $A$, $B$  and $C$ and edges correspond to $3$-sets in $D'$. Since $|D'|\geq \eta|A|^3$ applying Lemma \ref{hyperKST} on $G_3$ ($r=3$), we get a balanced complete $3$-partite $3$-graph in $G_3$ with color classes $A'$, $B'$ and $C'$, such that $|A'|=|B'|=|C'| \geq \beta\sqrt{\log |A|}$. Clearly $A'$, $B'$, $C'$ and a subset of $Z'_1$ of size $|A'|$, correspond to the color classes of the required complete $4$-partite $4$-graph.\hfill{} 
\end{proof}

\vskip4pt
We will frequently apply the following folklore statements (proofs are omitted):
\begin{lemma}\label{folklore_min_degree_subgraph}
 Every graph $H$ has a subgraph $H'$ such that $\delta_1(H') \geq |E(H)|/|V(H)|$.
\end{lemma}
\begin{lemma}\label{folklore_matching}
 Every $3$-graph $H$ on $n$ vertices, with $\delta_1(H) \geq \eta {n\choose 3}/n$, has a matching of size $\eta n/24$. 
\end{lemma}

\vskip4pt

\noindent Finally we use the following lemma from \cite{HPS_PM_3u_vertDeg}.
\vskip 3pt
\begin{lemma}\label{absorbLemma}(Absorbing Lemma) For every $\eta>0$, there is an integer $n_0 = n_0(\eta)$ such that if $H$ is a $4$-graph on  $n\geq n_0$ vertices with $\delta_1(H)\geq \left(1/2+2\eta\right){n\choose 3}$, then there exist a matching $M$ in $H$ of size $|M|\leq \eta^4n/4$ such that for every set $W\subset V\setminus V(M)$ of size at size at most $\eta^8 n\geq |W|\in 4\mathbb{Z}$, there exists a matching covering exactly the vertices in $V(M)\cup W$.
\end{lemma}

\section{Auxiliary results}\label{aux_results} 
For a $4\times4\times4$ $3$-graph, denote by $Q_1,Q_2,Q_3$ its $3$ color classes and let $Q_1 = \{a_1,a_2,a_3,a_4\}$, $Q_2 = \{b_1,b_2,b_3,b_4\}$ and $Q_3 = \{c_1,c_2,c_3,c_4\}$. For two vertices, $x\in Q_i,\; y\in Q_j$,\; $i\neq j$, $N(x,y)$ denotes the neighborhood of the pair $x,y$, i.e. the set of vertices in the third color class that make edges with the pair $x,y$. Let $deg(x,y)=|N(x,y)|$. For two disjoint pairs $(x_1,y_1)$ and $(x_2,y_2)$ in $Q_i\times Q_j$, $i\neq j$ the pairs $(x_1,y_2)$ and $(x_2,y_1)$ are called the crossing pairs and the value  $deg(x_1,y_2) + deg(x_2,y_1)$ is referred to as the {\em crossing degree sum}. We define the following four special $4\times4\times4$ $3$-graphs.
\vskip5pt
\begin{definition}
$H_{432}$ is a $4\times4\times4$ 3-graph, such that there exist $3$ disjoint pairs in a $Q_i\times Q_j$, $i\neq j$, with degrees at least $4$, $3$ and $2$ respectively.  
\end{definition}
In the figures the triplets joined by a line represent an edge in the $3$-graph. Let the $H_{432}$ be as in Figure \ref{fig_H_{432_1}} or Figure \ref{fig_H_{432_2}}. If the $H_{432}$ does not have a perfect matching then we must have $deg(b_4,c_4) = 0$, (as otherwise by the K\"onig-Hall criteria we get a perfect matching).
\begin{figure}[h!]\label{fig_H_432}
\begin{minipage}[b]{0.3\linewidth}
\centering
\subfigure[\footnotesize{$H_{432}$} with $|N(b_2,c_2) \cup N(b_3,c_3)| = 3$]{\label{fig_H_{432_1}}\includegraphics[width=2.54in]{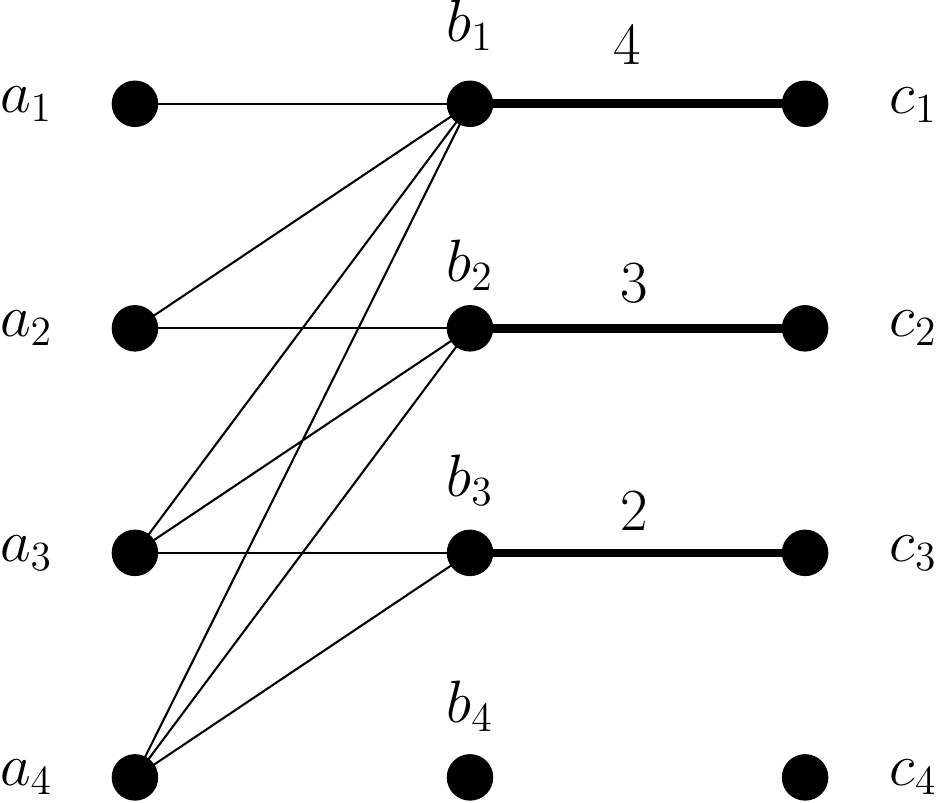}}
\end{minipage}
\hspace{3cm}
\begin{minipage}[b]{0.55\linewidth}
\centering
\subfigure[\footnotesize{$H_{432}$} with $|N(b_2,c_2) \cup N(b_3,c_3)| = 4$]{\label{fig_H_{432_2}}\includegraphics[width=2.54in]{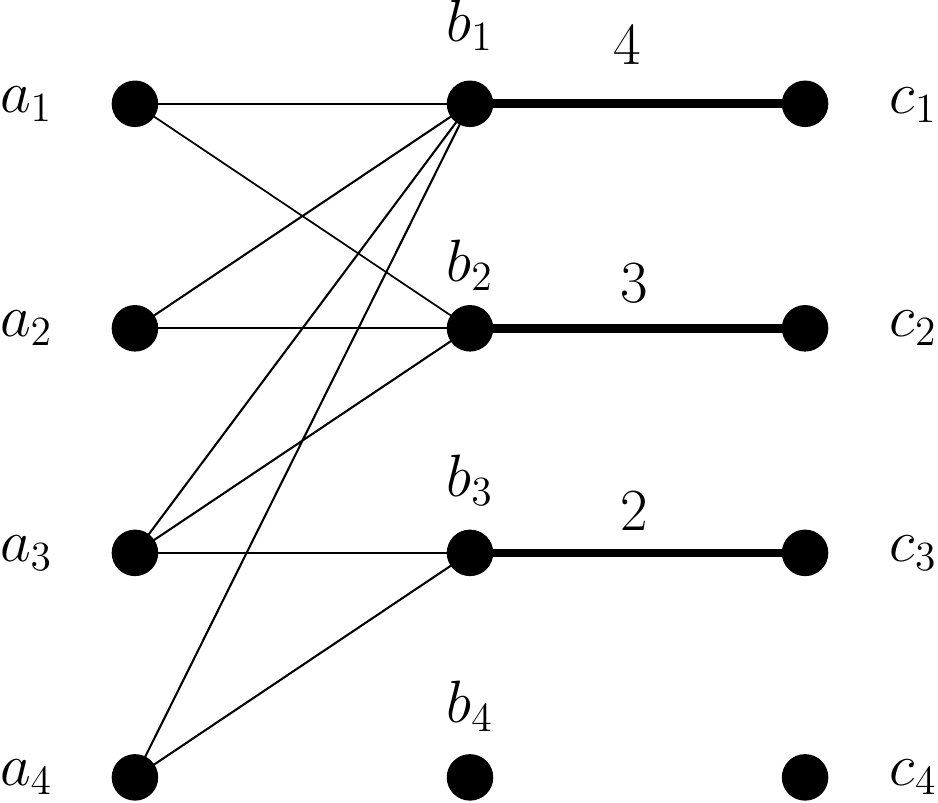}}
\end{minipage}
\caption{The two types of $H_{432}$. Labels on edges is the degree of the pair.}
\end{figure}
\begin{definition}
$H_{4221}$ is a $4\times4\times4$ tripartite 3-graph, such that there exist $4$ disjoint pairs in a $Q_i\times Q_j$, $i\neq j$, with degrees at least $4$, $2$, $2$ and $1$ respectively. 
\end{definition}
Let the $H_{4221}$ be as in Figure \ref{fig_H_{4221_1}}. If $|N(b_2,c_2) \cup N(b_3,c_3) \cup N(b_4,c_4)| = 3$ then again by the K\"onig-Hall criteria we get that $H_{4221}$ has a perfect matching. We only consider the $H_{4221}$ that has no perfect matching.

\begin{definition}
$H_{3321}$ is a $4\times4\times4$ 3-graph, such that there exist $4$ disjoint pairs in a $Q_i\times Q_j$, $i\neq j$, with degrees at least $3$, $3$, $2$ and $1$ respectively (see Figure \ref{fig_H_3221_1}).
\end{definition}

\begin{figure}[h!]\label{fig_H_4221_n_H_3321}
\begin{minipage}[b]{0.3\linewidth}
\centering
\subfigure[\footnotesize{$H_{4221}$ with no perfect matching}]{\label{fig_H_{4221_1}}\includegraphics[width=2.54in]{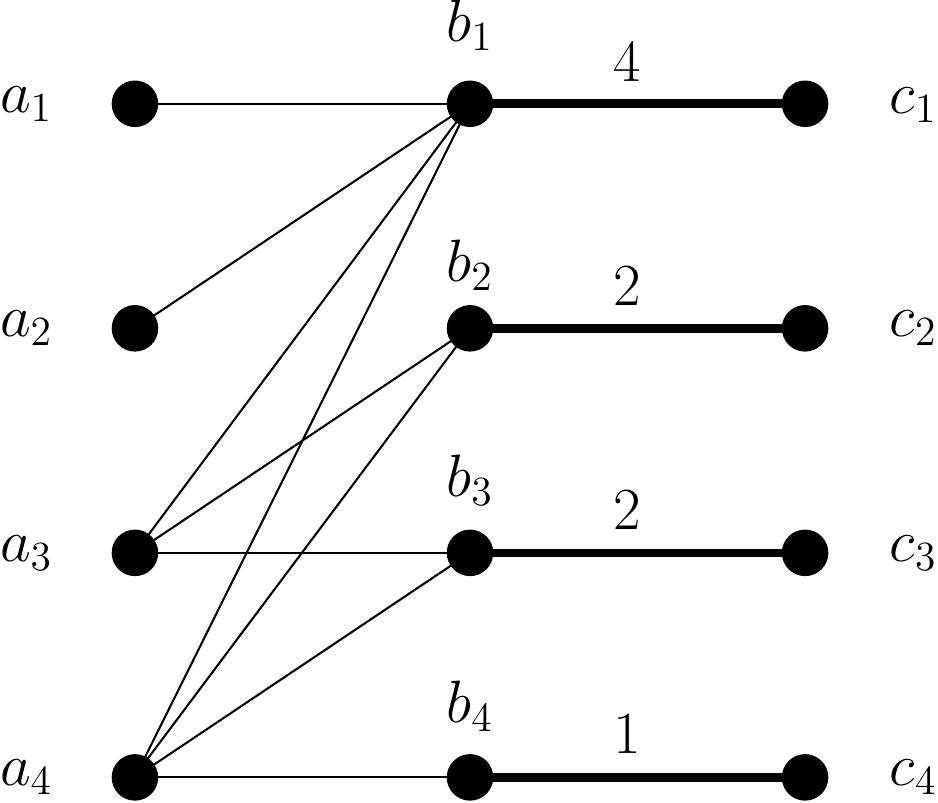}}
\end{minipage}
\hspace{3cm}
\begin{minipage}[b]{0.55\linewidth}
\centering
\subfigure[\footnotesize{$H_{3321}$ with no perfect matching}]{\label{fig_H_3221_1}\includegraphics[width=2.54in]{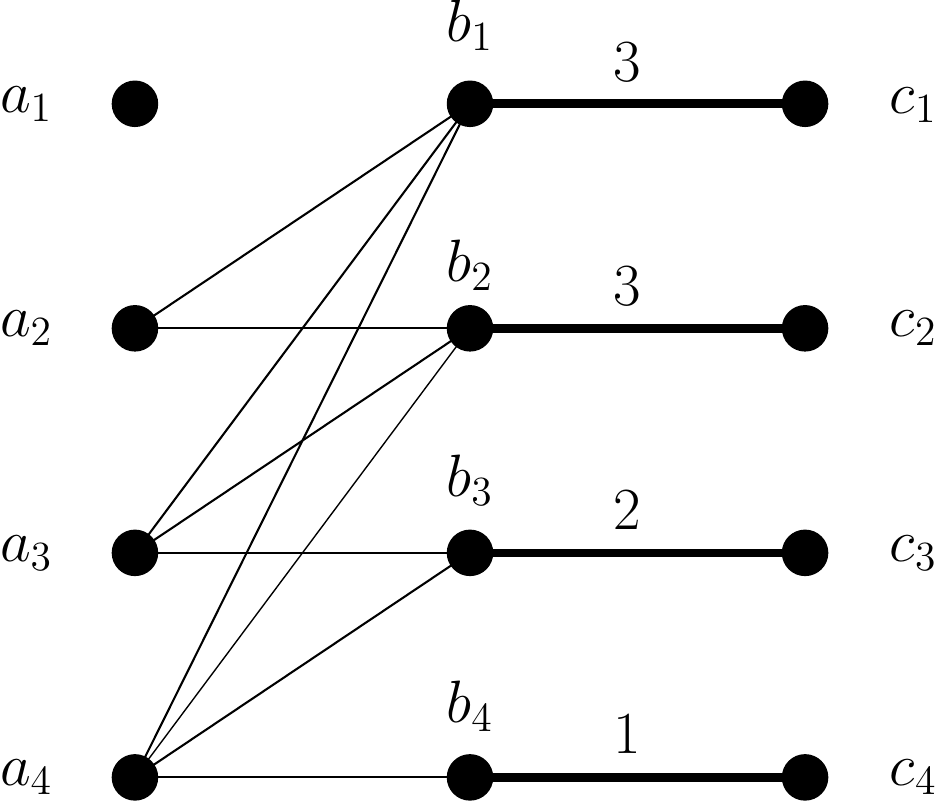}}
\end{minipage}
\end{figure}

\begin{definition}
$H_{ext}$ is a $4\times4\times4$ 3-graph with exactly $37$ edges such that there are  three vertices, one in each of $Q_1,Q_2$ and $Q_3$, and all edges are incident to at least one of these three vertices.
\end{definition}
\vskip6pt
The following lemma that classifies, $4\times4\times4$ $3$-graphs with at least $37$ edges, will be very useful in the subsequent section.  

\begin{lemma}\label{ext_categories}
Let $H(Q_1,Q_2,Q_3)$ be a $4\times4\times4$ $3$-graph. If $|E(H)|\geq 37$ then one of the following must be true
\begin{enumerate}
\item $H$ has a perfect matching.
\item $H$ has a subgraph isomorphic to $H_{3321}$
\item $H$ has a subgraph isomorphic to $H_{432}$
\item $H$ has a subgraph isomorphic to $H_{4221}$
\item $H$ is isomorphic to $H_{ext}$.
 
\end{enumerate}

\end{lemma}

\begin{proof}

We consider the following cases based on degree of pairs in $Q_i\times Q_j$. 

\vskip8pt\noindent \textbf{Case 1:} There is $Q_i$ and $Q_j$, $i\neq j$ such that no pair in $Q_i\times Q_j$ has degree $4$.
 
\vskip10pt\noindent Let $Q_2, Q_3$ be such a pair. Since $|E(H)|\geq 37$ and no pair has degree $4$, at least $5$ out of the $16$ pairs in $Q_2\times Q_3$ must be of degree $3$. Which implies that there must be at least $2$ disjoint pairs of degree $3$. Consider the largest set of disjoint pairs of degree $3$ in $Q_2\times Q_3$. Assume that there are $3$ disjoint pairs in $Q_2\times Q_3$ of degree $3$, say $(b_1,c_1)$, $(b_2,c_2)$ and $(b_3,c_3)$. 
\vskip6pt

If $deg(b_4,c_4)\geq 1$, then we have a $H_{3321}$. So assume that $deg(b_4,c_4)=0$, The total number of edges incident to pairs in $\{b_1,b_2,b_3\}\times\{c_1,c_2,c_3\}$ is at most $27$ (as there are $9$ pairs and degree of every pair is at most $3$), the remaining at least $10$ edges are incident to either $b_4$ or $c_4$. Which implies that there must be at least one pair (say $(b_3,c_3)$), such that crossing degree sum of $(b_3,c_3)$ and $(b_4,c_4)$ is at least $4$. Therefore we have that the degree of one crossing pair is at least $2$ and that of the other crossing pair is at least $1$. These two crossing pairs together with $(b_1,c_1)$ and $(b_2,c_2)$ gives us a subgraph isomorphic to $H_{3321}$. 
\vskip10pt

On the other hand, if there are exactly $2$ disjoint pairs of degree $3$, say $(b_1,c_1)$ and $(b_2,c_2)$. Again the total number of edges incident to pairs in $\{b_1,b_2\}\times\{c_1,c_2\}$ is at most $12$ (as there are $4$ pairs and degree of every pair is at most $3$). If there is a pair in $\{b_3,b_4\}\times\{c_3,c_4\}$ (say $(b_3,c_3)$) such that the {\em crossing degree sum} of $(b_1,c_1)$ and $(b_3,c_3)$ is at least $6$, then since there is no degree $4$ pair we must have that both $deg(b_1,c_3)$ and $deg(c_1,b_3)$ are $3$. Now these crossing pairs together with $(a_2,b_2)$ are $3$ disjoint pairs of degree $3$ which is a contradiction to the maximality of the set of disjoint pairs of degree $3$. Therefore we must have that the sum of degrees of pairs in $\{b_1,b_2\}\times \{c_3,c_4\}$ and $\{c_1,c_2\}\times \{b_3,b_4\}$ is at most $4\times 5 = 20$. Hence the number of edges of $H$ incident to pairs in $\{b_3,b_4\} \times \{c_3,c_4\}$ is at least $37-12-20 = 5$ and no pair has degree $3$. Therefore, in $\{b_3,b_4\}\times\{c_3,c_4\}$, we can find two disjoint disjoint pairs, (say $(b_3,c_3)$ and $(b_4,c_4)$) with degree at least $2$ and $1$ respectively and we get a graph isomorphic to $H_{3321}$.\\

\noindent \textbf{Case 2:} There is a $Q_i$ and $Q_j$, $i\neq j$ such that exactly one disjoint pair in $Q_i\times Q_j$ has degree $4$.
\vskip8pt\noindent
Let $Q_2,Q_3$ be such a pair. Consider the largest set of disjoint pairs in $Q_2\times Q_3$ with one pair of degree $4$ and the remaining of degree $3$. Note that if there are two disjoint pairs of degree $3$ besides the degree $4$ pair in the selected set, then clearly we have an $H_{432}$. So we consider the following two subcases based on whether or not there is a pair of degree $3$ in the selected set. Let $(b_1,c_1)$ be the degree $4$ pair in the selected set. 
\vskip8pt
\noindent \textbf{Subcase 2.1} There is another pair in $Q_2\times Q_3$ disjoint from $(b_1,c_1)$ with degree $3$. 
\vskip8pt
Assume that $deg(b_2,c_2) = 3$. First observe that
\begin{enumerate}
 \item If any pair in $\{b_3,b_4\}\times\{c_3,c_4\}$ has degree at least $2$ then that pair together with $(b_1,c_1)$ and $(b_2,c_2)$ makes an $H_{432}$.

\item If both $deg(b_1,c_2)$ and $deg(b_2,c_1)$ are $4$ then we get two disjoint degree $4$ pairs. Therefore we have that the number of edges incident to pairs in $\{b_1,b_2\}\times\{c_1,c_2\}$ is at most $4+3+4+3 = 14$.

\item If there is a pair in $\{b_3,b_4\}\times\{c_3,c_4\}$ (say $(b_3,c_3)$) such that the {\em crossing degree sum} of $(b_2,c_2)$ and $(b_3,c_3)$ is at least $5$. Then we must have that one crossing pair is of degree at least $3$, and the other is of degree at least $2$ (because none of them can be of degree $4$). These $2$ crossing pairs together with $(b_1,c_1)$ makes the disjoint pairs of an $H_{432}$. Therefore we must have that $(b_2,c_2)$ and any pair in $\{b_3,b_4\}\times\{c_3,c_4\}$ have their {\em crossing degree sum} at most $4$.

\item Similarly $(b_1,c_1)$ and any pair in $\{b_3,b_4\}\times\{c_3,c_4\}$ have their {\em crossing degree sum} at most $6$. 
\end{enumerate}

\noindent Assume that $(b_1,c_1)$ and $(b_3,c_3)$ have their crossing degree sum, equal to $6$. If the degrees of crossing pairs are $4$ and $2$, then these crossing pairs and $(b_2,c_2)$ makes the disjoint pairs of an $H_{432}$. On the other hand if both the crossing pairs have degree $3$. Then $(b_1,c_3)$, $(b_3,c_1)$ and $(b_2,c_2)$ are three disjoint pairs of degree $3$. From observation 1 we have $deg(b_3,c_3)\leq 1$ and from observation 3 the crossing degree sum of $(b_3,c_3)$ and $(b_2,c_2)$ is at most $4$. Which together with observation 2 gives us that the total number of edges incident to pairs in $\{b_1,b_2,b_3\}\times\{c_1,c_2,c_3\}$ is at most $14+6+1+4=25$. Now if $deg(b_4,c_4)=1$ then we have an $H_{3321}$, otherwise we have that the number of edges containing either $b_4$ or $c_4$ is at least $37-25\geq 12$. By observation 1 we have that both $deg(b_4,c_3)$ and $deg(b_3,c_4)$ are at most $1$  hence we must have that the {\em crossing degree sum} of $(b_2,c_2)$ and $(b_4,c_4)$ is $4$ with one crossing pair of degree at least $1$ and the other of degree at least $2$. These crossing pairs together with $(b_1,c_3)$ and $(b_3,c_1)$ gives us an $H_{3321}$.
\vskip10pt
\noindent On the other hand if for any pair in $\{b_3,b_4\}\times\{c_3,c_4\}$ and $(b_1,c_1)$ their {\em crossing degree sum} is be at most $5$. Then the number of edges incident to pairs in $\{b_3,b_4\}\times\{c_3,c_4\}$ is at least $37-14-4(2)-5(2) =5$. Which implies that there must a degree $2$ pair, and hence by observation 1, we get an $H_{432}$.\\ 

\noindent\textbf{Subcase 2.2} There is no pair of degree $3$ disjoint from $(b_1,c_1)$.
\vskip6pt
In this case again as in observation 2, the {\em crossing degree sum} of $(b_1,c_1)$ and any pair in $\{b_2,b_3,b_4\}\times\{c_2,c_3,c_4\}$ is at most $6$ (as any other case results in two disjoint pairs of degree $4$ and $3$). This implies that the number of edges incident to pairs in $\{b_2,b_3,b_4\}\times\{c_2,c_3,c_4\}$ is at least $37- 4 - 3(6) = 15$ and no pair has degree $3$. Which implies that there are three disjoint pairs in $\{b_2,b_3,b_4\}\times\{c_2,c_3,c_4\}$ with degrees $2$, $2$ and at least $1$ respectively. These pairs and $(b_1,c_1)$ makes the $4$ disjoint pairs of an $H_{4221}$.\\ 

\noindent \textbf{Case 3:} In every $Q_i$ and $Q_j$, $i\neq j$ there are exactly two disjoint pairs in $Q_i\times Q_j$ with degree $4$.\\

Consider $Q_1,Q_2$ and assume that $(a_1,b_1)$ and $(a_2,b_2)$ are the two disjoint pairs with degree $4$. We make the following observations:
\begin{enumerate}
\item If any pair in $\{a_3,a_4\}\times\{b_3,b_4\}$ has degree at least $2$ then that pair together with $(a_1,b_1)$ and $(a_2,b_2)$ makes the disjoint pairs of an $H_{432}$. Therefore the total number of edges spanned by pairs in  $\{a_3,a_4\}\times\{b_3,b_4\}$ is at most $4$.

\item For any of $(a_1,b_1)$ and $(a_2,b_2)$ and any pair in $\{a_3,a_4\}\times\{b_3,b_4\}$ their {\em crossing degree sum} can be at most $5$. Indeed otherwise say the {\em crossing degree sum} of $(a_1,b_1)$ and $(a_3,b_3)$ is $6$, then we must have that one crossing pair has degree at least $3$, and the other has degree at least $2$. These crossing pairs and $(a_2,b_2)$ make the disjoint pairs of an $H_{432}$. Furthermore if any such {\em crossing degree sum} is $5$ then by the same reasoning as above, it must be that one crossing pair is of degree $4$ and the other is of degree $1$. 


\item If the total number of edges spanned by pairs in $\{a_1,a_2\}\times\{b_1,b_2\}$ is at most $12$, then the number of edges that uses one vertex from $\{a_1,a_2,b_1,b_2\}$ and one vertex from $\{a_3,a_4,b_3,b_4\}$ is at least $37-12-4 =21$. Hence there will be a pair in $\{a_3,a_4\}\times\{b_3,b_4\}$ (say $(a_3,b_3)$) such that {\em crossing degree sum} of $(a_3,b_3)$ and at least one of $(a_1,b_1)$ and $(a_2,b_2)$ is at least $6$, and by observation 2 we get an $H_{432}$. 

\end{enumerate}
\vskip 4pt
The above observations are true for any two disjoint pairs of degree $4$. We choose two disjoint pairs of degree $4$, (say $(a_1,b_1)$ and $(a_2,b_2)$) pairs in $Q_1\times Q_2$ such that (i) $a_1$ has the maximum vertex degree among all vertices that are part of some pairs of degree $4$ and (ii) $deg(a_2,b_1)$ is as small as possible. Let $deg(a_2,b_1) = x$. Note that by observation 3, we have $x\geq 1$. We consider the following cases based on the value of $x$. 



\vskip8pt
\noindent\textbf{Subcase 3.1} $x= 4$
\vskip8pt
Note that that the number of edges in $\{a_1,a_2\} \times \{b_1,b_2\}$ is at most $16$ ($deg(a_1,b_2)\leq 4$). First we will show that both $deg(a_3,b_1),deg(a_4,b_1)\leq 1$. Assume that $deg(a_3,b_1)\geq 2$, but then as in observation 2 both $deg(a_1,b_3)$ and $deg(a_1,b_4)$ can be at most $2$. Hence by the maximality of the degree of $a_1$ we get $deg(a_3,b_1)+deg(a_4,b_1) \leq 4$. Now by observation 1, there must be at least $37-16-4-8 =9$ edges containing one vertex from $\{a_3,a_4,b_3,b_4\}$ and one of $\{a_2,b_2\}$. Which implies that the degree of $a_2$ or $b_2$ is strictly larger than that of $a_1$, a contradiction. So we have $deg(a_3,b_1)=deg(a_4,b_1)=1$. 
\vskip8pt
Now we show that both $deg(a_3,b_2),deg(a_4,b_2) \leq 1$. To see this first assume that either $deg(a_1,b_3)$ or $deg(a_1,b_4)$ is equal to $4$, (say $deg(a_1,b_3)=4$) then by the minimality of $x$ we must have that $deg(a_2,b_3)=4$ too, because if $deg(a_2,b_3)<x$, then we can exchange $b_1$ with $b_3$ to get a smaller value of $x$. But if  $deg(a_2,b_3)=4$ then by observation 2 we must have $deg(a_3,b_2)=deg(a_4,b_2)=1$. On the other hand if both $deg(a_1,b_3)$ and $deg(a_1,b_4)$ are at most $3$ ($deg(a_1,b_3)+deg(a_1,b_4) \leq 6$) then again there must be at least $37-16-4-8 =9$ edges containing some pair in $\{a_2\}\times \{b_3,b_4\}$ and $\{a_3,a_4\}\times \{b_2\}$. Which means at least one of these pairs must be of degree at least $3$. Say $deg(a_2,b_3)\geq 3$, but then by observation 2 we have $deg(a_3,b_2),deg(a_4,b_2) \leq 1$ and we are done. In case say $deg(a_3,b_2)\geq 3$ then we have $deg(a_3,b_2)+deg(a_4,b_2) \geq 7$ and we get that the degree of $b_2$ is larger than degree of $a_1$, a contradiction. 
\vskip8pt
So we have that $deg(a_3,b_1),deg(a_4,b_1),deg(a_3,b_2)\mbox{ and } deg(a_4,b_2) \leq 1$. This together with observation 1 implies that the number of edges containing some pair in $\{a_1,a_2\} \times \{b_1,b_2,b_3,b_4\}$ is at least $37-4-4 = 29$. Therefore the $2\times4\times4$ $3$-graph $H'(\{b_1,b_2\},Q_2,Q_3)$ has at least $37-4-4=29$ edges. There must be at least three vertices in $Q_2$ such that each one of them is part of at least three pairs in $(Q_2\times Q_3)$ that are of degree at least $2$. To see this assume that there are at most two such vertices in $Q_2$ (say $b_1$ and $b_2$). Then using the fact that the maximum degree of a pair in $Q_2\times Q_3$ in $H'$ is $2$, we get that $b_1$ and $b_2$ can be contained in at most $2(4\cdot 2) = 16$ edges. While at most $2$ pairs containing either $b_3$ and $b_4$ can be of degree $2$, we get that the number of edges containing either $b_3$ and $b_4$ is at most $2\cdot(2\cdot 2  + 2\cdot 1) = 12$ which implies that $|E(H')|\leq 28$ a contradiction 
\vskip 4pt
Now since in $Q_2$ there are at least $3$ vertices such that each one of them is part of at least $3$ pairs in $(Q_2\times Q_3)$ of degree at least $2$. Which implies that there must a degree $2$ pair disjoint from the two degree $4$ pairs (guaranteed in Case 3) in $Q_2\times Q_3$. Therefore we get an $H_{432}$.  

\vskip8pt
\noindent\textbf{Subcase 3.2} $x= 3$
\vskip8pt
\noindent Now we have that the number of edges in $\{a_1,a_2\} \times \{b_1,b_2\}$ is at most $15$. Again we first show that both $deg(a_3,b_1),deg(a_4,b_1)\leq 1$. Assume that $deg(a_3,b_1)\geq 2$, but then by observation 2 both $deg(a_1,b_3)$ and $deg(a_1,b_4)$ can be at most $2$. Hence by the maximality of the degree of $a_1$ we get $deg(a_3,b_1)+deg(a_4,b_1) \leq 5$. Now by observation 1, there must be at least $37-15-4-9 =9$ edges containing one vertex from $\{a_3,a_4,b_3,b_4\}$ and one of $\{a_2,b_2\}$. Which implies that the degree of $a_2$ or $b_2$ is strictly larger than that of $a_1$, a contradiction. So we have $deg(a_3,b_1)=deg(a_4,b_1)=1$. 
\vskip8pt
\noindent Similarly as in the previous case we show that both $deg(a_3,b_2),deg(a_4,b_2) \leq 1$. To see this first assume that either $deg(a_1,b_3)$ or $deg(a_1,b_4)$ is equal to $4$, (say $deg(a_1,b_3)=4$) then by the minimality of $x$ we must have that $deg(a_2,b_3)\geq 3$ too. But if $deg(a_2,b_3)=3$ then by observation 2 we must have $deg(a_3,b_2),deg(a_4,b_2)\leq1$. On the other hand if both $deg(a_1,b_3)$ and $deg(a_1,b_4)$ are at most $3$ ($deg(a_1,b_3)+deg(a_1,b_4) \leq 6$) then again there must be at least $37-15-4-8 =9$ edges containing some pair in $\{a_2\}\times \{b_3,b_4\}$ and $\{a_3,a_4\}\times \{b_2\}$. Which means at least one of these pairs must be of degree at least $3$. Say $deg(a_2,b_3)\geq 3$, but then by observation 2 we have $deg(a_3,b_2),deg(a_4,b_2) \leq 1$ and we are done. In case say $deg(a_3,b_2)\geq 3$ then we have $deg(a_3,b_2)+deg(a_4,b_2) \geq 7$ and we get that the degree of $b_2$ is greater than the degree of $a_1$, a contradiction. 
\vskip8pt
So we have that $deg(a_3,b_1),deg(a_4,b_1),deg(a_3,b_2)\mbox{ and } deg(a_4,b_2) \leq 1$. Again we get that the $2\times4\times4$ $3$-graph $H'(\{b_1,b_2\},Q_2,Q_3)$ has at least $37-4-4=29$ edges. and we are done. 

\vskip8pt
\noindent\textbf{Subcase 3.3} $x= 2$
\vskip8pt
\noindent Similarly as in the previous two subcases we have that $deg(a_3,b_1),deg(a_4,b_1),deg(a_3,b_2)\mbox{ and } deg(a_4,b_2) \leq 1$ and  the $2\times4\times4$ $3$-graph $H'(\{b_1,b_2\},Q_2,Q_3)$ has at least $37-4-4=29$ edges and we are done. 

\vskip8pt
\noindent\textbf{Subcase 3.4} $x= 1$
\vskip8pt

In this case observation 3 implies that the number of edges in $\{a_1,a_2\} \times \{b_1,b_2\}$ is exactly $13$ ($deg(a_1,b_2)=4$). Using $|E(H)\geq 37|$ and  observation $2$ we get that every pair in $\{a_3,a_4\}\times\{b_3,b_4\}$ has degree exactly $1$ and for any pair in $\{a_3,a_4\}\times\{b_3,b_4\}$ and any of $(a_1,b_1)$ or $(a_2,b_2)$ their {\em crossing degree sum} is exactly $5$ ($4+1$). 
\vskip8pt

\noindent Therefore, we have that either  $$deg(a_1,b_3)=deg(a_1,b_4)=deg(a_3,b_2)=deg(a_4,b_2)=4\;\;\; \mbox{     or}$$ $$deg(a_1,b_3)=deg(a_1,b_4)=deg(a_2,b_3)=deg(a_2,b_4)=4\;\;\;\;\;\;\;$$ 
\vskip8pt
\noindent In the latter case note that again we have that $H'(\{b_1,b_2\},Q_2,Q_3)$ has at least $29$ edges and we are done as above. 
\vskip8pt
\noindent So assume that $deg(a_1,b_3)=deg(a_1,b_4)=deg(a_3,b_2)=deg(a_4,b_2)=4$ and $deg(a_1,b_1) = deg(a_2,b_2) = deg(a_1,b_2) = 4$ and every other pair in $Q_1\times Q_2$ is of degree exactly $1$. This means that $a_1$ and $b_2$ are not part of any degree $1$ pair. Now the neighborhoods of all the degree $1$ pairs must be the same vertex in $Q_3$ (say $c_3$). Because otherwise we get a perfect matching in $H$, (using those two vertices in $Q_3$ for two disjoint degree $1$ pairs, and the remaining two vertices of $Q_3$ are matched with two of the degree $4$ pairs in $(Q_1\times Q_2)$). But if all of these edges are incident to $c_3$, then all edges in this graph are incident to at least one vertex in $\{a_1,b_2,c_3\}$ and the number of edges is exactly $37$, hence $H$ is isomorphic to $H_{ext}$. \hfill{ } 
\end{proof}

\vskip 16pt

\noindent Let $A$, $B$ and $C$ be three disjoint balanced complete $4$-partite $4$-graphs with color classes  $(A_1,\ldots,A_4)$, $(B_1,\ldots,B_4)$ and $(C_1,\ldots,C_4)$ respectively, and $|A_1|= |B_1|= |C_1| = m$. Let $Z$ be a set of vertices disjoint from vertices in $A$, $B$ and $C$. For a small constant $\eta>0$, we say that $Z$ is {\em connected} to a triplet of color classes $(A_i,B_j,C_k)$, $1\leq i,j,k \leq 4$, if $d_4(Z,(A_i \times B_j\times C_k)) \geq 2\eta$. For $Z$ and $(A,B,C)$ we define an auxiliary graph, (the {\em link graph}), $L_{abc}$ to be a $4\times4\times4$ $3$-graph where the vertex set of each color class of $L_{abc}$ corresponds to the color classes in $A$, $B$ and $C$. While a triplet of vertices $(a_i,b_j,c_k)$ is an edge in $L_{abc}$ iff $Z$ is connected to the triplet of color classes $(A_i,B_j,C_k)$. 
\vskip 10pt
\noindent Given three balanced complete $4$-partite $4$-graphs $A$, $B$ and $C$ and another set of vertices $Z$, as above, we say that we can {\em extend} $(A,B,C)$ if we can build another set of balanced complete $4$-partite $4$-graphs using  $V(A)\cup V(B) \cup V(C) \cup Z$ such that the total number of vertices in the new $4$-partite $4$-graphs is at least $12m + \eta m/16$ and the size of a color class in each new $4$-partite $4$-graph is $\beta\sqrt{\log m}$. In what follows we outline a procedure to extend $(A,B,C)$ using the structure of $L_{abc}$.


\begin{lemma}\label{extensionLemma}
For $\eta,c >0$, let $A,B$ and $C$ be three balanced complete $4$-partite $4$-graphs such that $|A_1|= |B_1|= |C_1| = m$. If $Z$ is a disjoint set of vertices with $|Z|\geq c 2^{m^3}$. If the link graph $L_{abc}$ has at least $37$ edges and $L_{abc}$ is not isomorphic to $H_{ext}$ then we can extend $(A,B,C)$. 
\end{lemma}

\begin{proof}
Since $L_{abc}$ is a $4\times4\times4$ $3$-graph with at least $37$ edges and is not isomorphic to $H_{ext}$, for each of the other cases as in Lemma \ref{ext_categories}, we give the procedure to extend $(A,B,C)$. 

\vskip6pt 
\noindent {\bf Case 1:} $L_{abc}$ has a perfect matching:

\vskip3pt 
\noindent Without loss of generality assume that the perfect matching in $L_{abc}$ corresponds to $\{(A_i,B_i,C_i)\; : 1\leq i\leq 4\}$ i.e. $Z$ is connected to the triplets $\{(A_i,B_i,C_i)\; : 1\leq i\leq 4\}$. Note that by the definition of  connectedness and the sizes of the sets the $4$-partite $4$-graph $(A_1,B_1,C_1,Z)$ satisfies the conditions of Lemma \ref{4partVolArg}. Hence we find a complete balanced $4$-partite $4$-graph $X_1 = (A_1^1,B_1^1,C_1^1,Z^1)$, such that $$ Z^1 \subset Z,\;\; {A}_1^1 \subset A_1 \mbox{  ,  } {B}_1^1 \subset B_1 \mbox{  and  } C_1^1 \subset C_1  \;\; \mbox{ and }$$ $$|Z^1|=|A_1^1| = |B_1^1|=|C_1^1| = \beta\sqrt{\log m} \;\;\mbox{  where } \beta \mbox{ is as in Lemma \ref{4partVolArg}}$$

\begin{figure}[h!] 
\centering
\includegraphics[scale=.42]{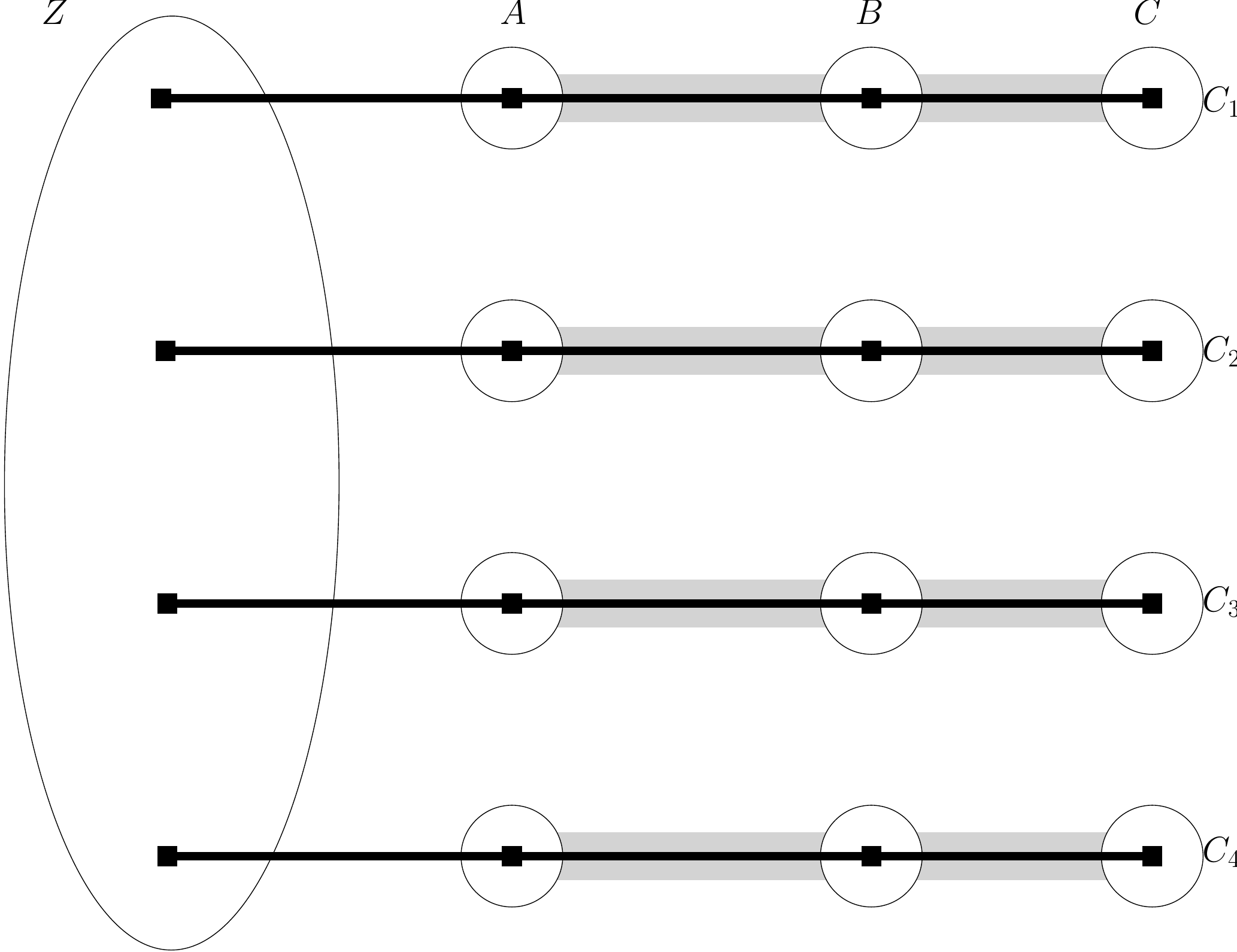} 
\caption{\footnotesize{Extending $(A,B,C)$ when $L_{abc}$ has a perfect matching: The shaded boxes represent a triplet {\em connected} to $Z$, while solid thick lines represent a balanced complete $4$ partite graphs}}
\label{fig_4_unif_PM_extended}
\end{figure}

\noindent Similarly, we find such complete balanced $4$-partite $4$-graphs $X_2$, $X_3$ and $X_4$ in $(A_2,B_2,C_2,Z)$, $(A_3,B_3,C_3,Z)$ and $(A_4,B_4,C_4,Z)$ respectively, that are disjoint from each other (as $|Z|$ is very large compared to $m$) (see Figure \ref{fig_4_unif_PM_extended}). We remove the vertices in $X_1,X_2,X_3$ and $X_4$ and make these four new $4$-partite $4$-graphs. In the remaining parts of $A,B$ and $C$ we remove another such set of $4$ disjoint complete balanced $4$-partite $4$-graphs. Again by definition of connectedness and Lemma \ref{4partVolArg} we can continue this process until we remove at least $\eta m/8$ vertices from each color class of $A,B$ and $C$.
\vskip6pt
\noindent Note that the new $4$-partite $4$-graphs use at least $4\eta m/8$ vertices from $Z$. Therefore these new $4$-partite $4$-graphs together with leftover parts of $A,B$ and $C$ have at least $3(4m) + \eta m/2$ vertices while all the $4$-partite $4$-graphs are balanced. Hence, we extended $(A,B,C)$. 

\vskip8pt \noindent {\bf Case 2:} $L_{abc}$ has a subgraph isomorphic to $H_{432}$: 

\vskip4pt \noindent In this case we show in detail how to extend such an $(A,B,C)$, while in the latter cases we will only briefly outline the procedure. First assume that the $H_{432}$ in $L_{abc}$ is as in Figure \ref{fig_H_{432_1}} and let the pairs corresponding to the degree $4$, $3$ and $2$ pairs in this subgraph be $(B_1,C_1)$, $(B_2,C_2)$ and $(B_3,C_3)$ respectively.  Furthermore let the color classes corresponding to the neighbors of degree $3$ and degree $2$ pairs in $H_{432}$ be $\{A_2,A_3,A_4\}$ and $\{A_3,A_4\}$ respectively. 
\vskip3pt
\noindent Using the definition of connectedness and Lemma \ref{4partVolArg} we find two disjoint complete balanced $4$-partite $4$-graphs $X_1 = (A_4^1,B_3^1,C_3^1,Z^1)$ and $X_2 = (A_3^2,B_3^2,C_3^2,Z^2)$ such that $ Z^j, A_i^j,B_i^j,C_i^j$ are subsets of $Z,A_i,B_i,C_i$ respectively. 
\begin{figure}[h!] 
\centering
\includegraphics[scale=0.58]{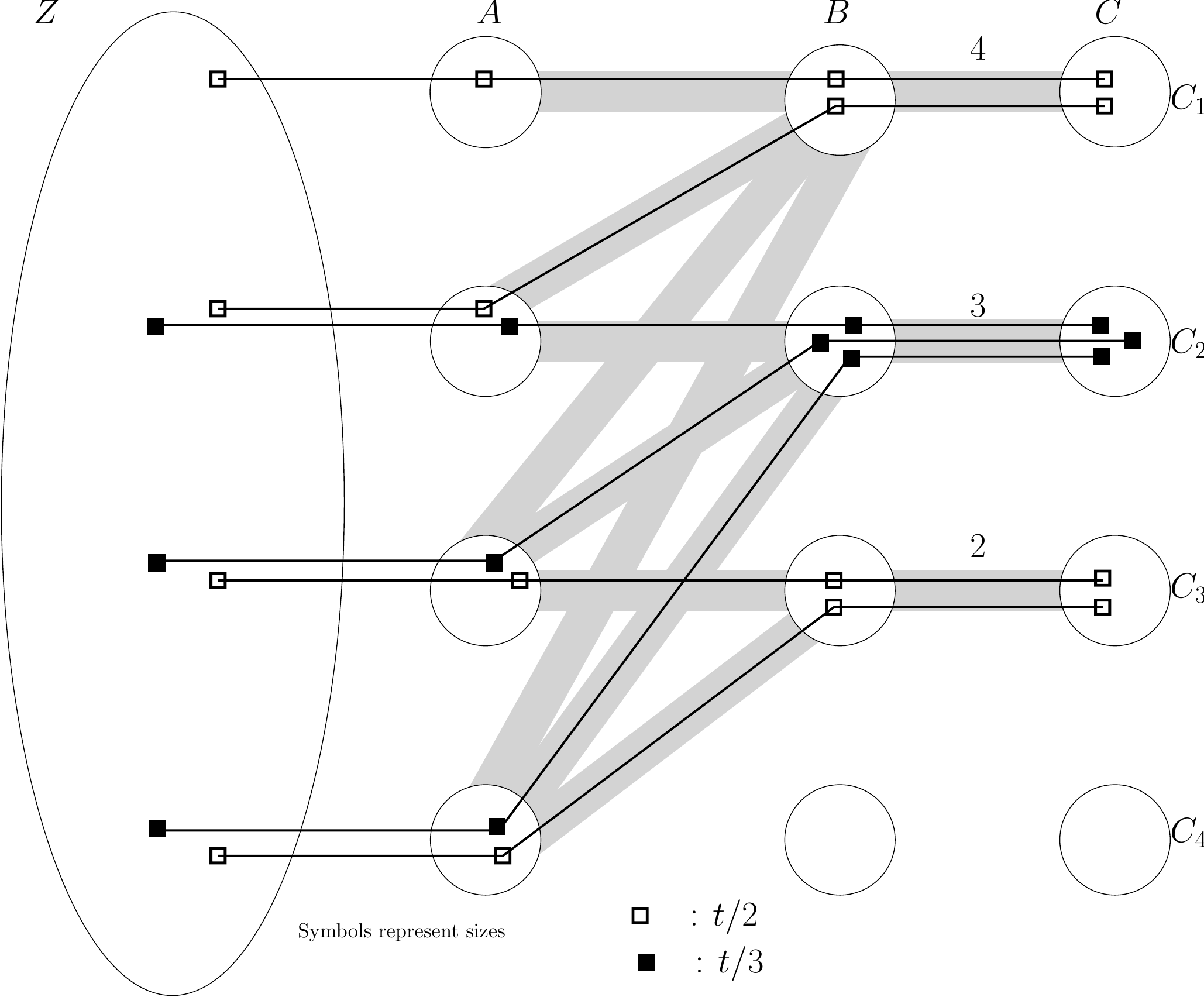} 
\caption{\footnotesize{Extending $(A,B,C)$ when $L_{abc}$ has an $H_{432}$ as in Figure \ref{fig_H_{432_1}}: The shaded boxes represent a triplet {\em connected} to $Z$, while solid think lines represent a balanced complete $4$ partite $4$-graphs}}
\label{fig_4_unif_H_432_1_extended}
\end{figure}

\noindent Similarly we find three more disjoint balanced complete $4$-partite $4$-graphs $X_3,X_4$ and $X_5$ where vertices of three color classes in all of them are from $B_2,C_2$ and $Z$ while vertices of the fourth color class are from $A_2,A_3$ and $A_4$ respectively. We build two more disjoint balanced complete $4$-partite $4$-graphs $X_6$ and $X_7$ such that vertices of three color classes in both of them are from $B_1,C_1$ and $Z$ while vertices of the fourth color class are from $A_1$ and $A_2$ respectively. 
\vskip 3pt

\noindent The size of a color class in $X_1,\ldots,X_7$ is $\beta\sqrt{\log m}$. We remove the vertices in $X_1,\ldots X_7$ from their color classes to make these $7$ new $4$-partite $4$-graphs
\vskip 6pt
In the remaining parts of $A,B$ and $C$ we remove another such set of seven disjoint balanced complete $4$-partite $4$-graphs that are disjoint from the previous ones. Again by definition of connectedness and Lemma \ref{4partVolArg} we can continue this process until we remove $\eta m/8$ vertices each from $B_i$ and $C_i$, $1\leq i\leq 3$. By construction, if the number of vertices used from $B_i$ and $C_i$, $1\leq i\leq 3$ is $t$ ($=\eta m/8$) then the number of vertices used in $A_2,A_3$ and $A_4$ is $5t/6$, while that in $A_1$ is $t/2$ (see Figure \ref{fig_4_unif_H_432_1_extended}).
\vskip6pt
\noindent Note that the new $4$-partite $4$-graphs use at least $3t \geq 3\eta m/8$ vertices from $Z$, but the remaining parts of $A$, $B$ and $C$ are not balanced ($A_1$, $B_4$ and $C_4$ have more vertices). To restore the balance in the remaining part of $A$ we discard some arbitrary $t/3$ vertices from the remaining part of $A_1$. Similarly we discard some arbitrary $t$ vertices from $B_4$ and $C_4$ to restore the balance in the remaining part of $B$ and $C$. Therefore the new $4$-partite $4$-graphs together with leftover parts of $A,B$ and $C$ (after discarding the vertices) have at least $4(|A_1|+|B_1|+|C_1|) + 3t - t/3 - 2t \geq 12m + \eta m/12$ vertices while all the $4$-partite $4$-graphs are balanced. Hence we extended $(A,B,C)$.
\begin{figure}[h!] 
\centering
\includegraphics[scale=0.58]{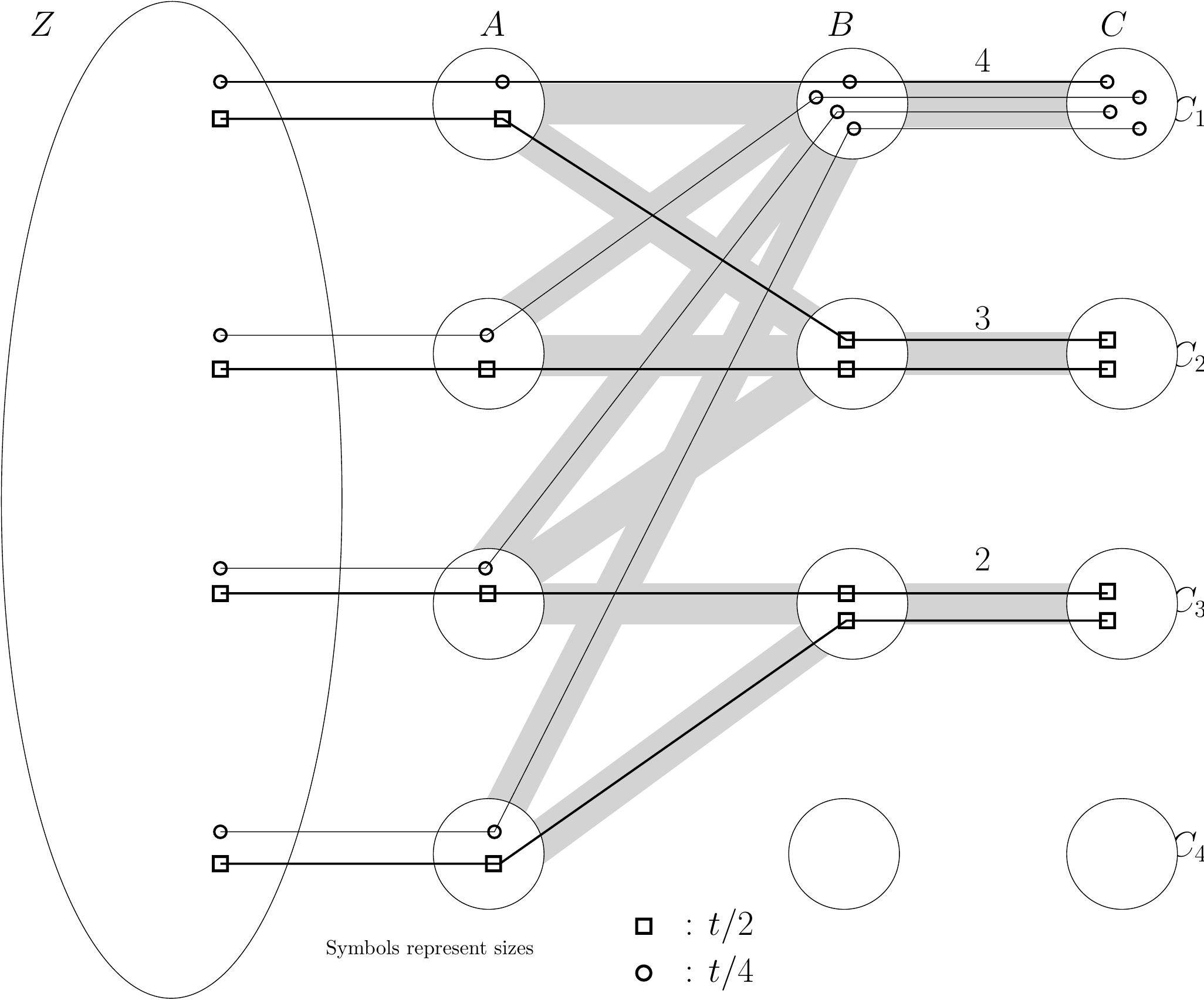} 
\caption{\footnotesize{Extending $(A,B,C)$ when $L_{abc}$ has an $H_{432}$ as in Figure \ref{fig_H_{432_2}}: The shaded boxes represent a triplet {\em connected} to $Z$, while solid thick lines represent a balanced complete $4$ partite $4$-graphs}}
\label{fig_4_unif_H_432_2_extended}
\end{figure}

\vskip2pt
\noindent On the other hand if the $H_{432}$ in $L_{abc}$ is as in Figure \ref{fig_H_{432_2}}, then let the color classes corresponding to the neighbors of degree $3$ and degree $2$ pairs in $H_{432}$ be $\{A_1,A_2,A_3\}$ and $\{A_3,A_4\}$ respectively. In this case extend $(A,B,C)$ as follows.

\vskip3pt 
\noindent For the pair $(B_3,C_3)$ we remove two balanced complete $4$-partite $4$-graphs with the fourth color classes in $A_3$ and $A_4$ respectively. For the pair $B_2,C_2$ we remove two balanced complete $4$-partite $4$-graphs with the fourth color classes in $A_1$ and $A_2$ respectively. The size of of each color class in all of these new $4$-partite graphs is $\beta\sqrt{\log m}$. Since the pair $(B_1,C_1)$ has degree $4$, we remove four balanced complete $4$-partite $4$-graphs with the fourth color class in $A_1,A_2,A_3$ and $A_4$ respectively.

\vskip3pt 
Similarly as in the previous case we repeat this process so that we remove at least $t\geq\eta m/8$ vertices from each $B_i$ and $C_i$, $1\leq i\leq 3$. Note that by construction we have used $3t/4$ vertices in each color class of $A$ (see Figure \ref{fig_4_unif_H_432_2_extended}). So the remaining part of $A$ is still balanced. While to restore balance in the remaining parts of $B$ and $C$, we discard some arbitrary $t$ vertices from each of $B_4$ and $C_4$. Again in total we added $3t$ vertices from $Z$, while we discarded $2t$ vertices form $B_4$ and $C_4$. Therefore the net increase in the number of vertices in the new set of complete $4$-partite $4$-graphs is $t\geq\eta m/8$, while all the $4$-partite $4$-graphs are balanced.

\noindent {\bf Case 3:} $L_{abc}$ has a subgraph isomorphic to $H_{4221}$: 

\vskip2pt 
\noindent Without loss of generality, assume that the pairs corresponding to the degree $4$, $2$, $2$ and $1$ pairs in this $H_{4221}$ are $(B_1,C_1)$, $(B_2,C_2)$, $(B_3,C_3)$ and $(B_4,C_4)$ respectively.  Furthermore let the color classes corresponding to the neighbors of degree $1$ and the two degree $2$ pairs in $H_{4221}$ be $\{A_4\}$, $\{A_3,A_4\}$ and $\{A_3,A_4\}$ respectively (as in Figure \ref{fig_H_{4221_1}}). By the definition of connectedness and Lemma \ref{4partVolArg} we build complete balanced $4$-partite $4$-graphs using $(A_4,B_4,C_4,Z)$ and $(A_3,B_3,C_3,Z)$ of size $\beta\sqrt{\log m}$. For the pair $(B_2,C_2)$ we make two more balanced complete  $4$-partite $4$-graphs using $(A_4,B_2,C_2,Z)$ and $(A_3,B_2,C_2,Z)$. For the degree $4$ pair $(B_1,C_1)$ we remove balanced complete  $4$-partite $4$-graphs in $(A_1,B_1,C_1,Z)$ and $(A_2,B_1,C_1,Z)$.
\begin{figure}[h!] 
\centering
\includegraphics[scale=0.58]{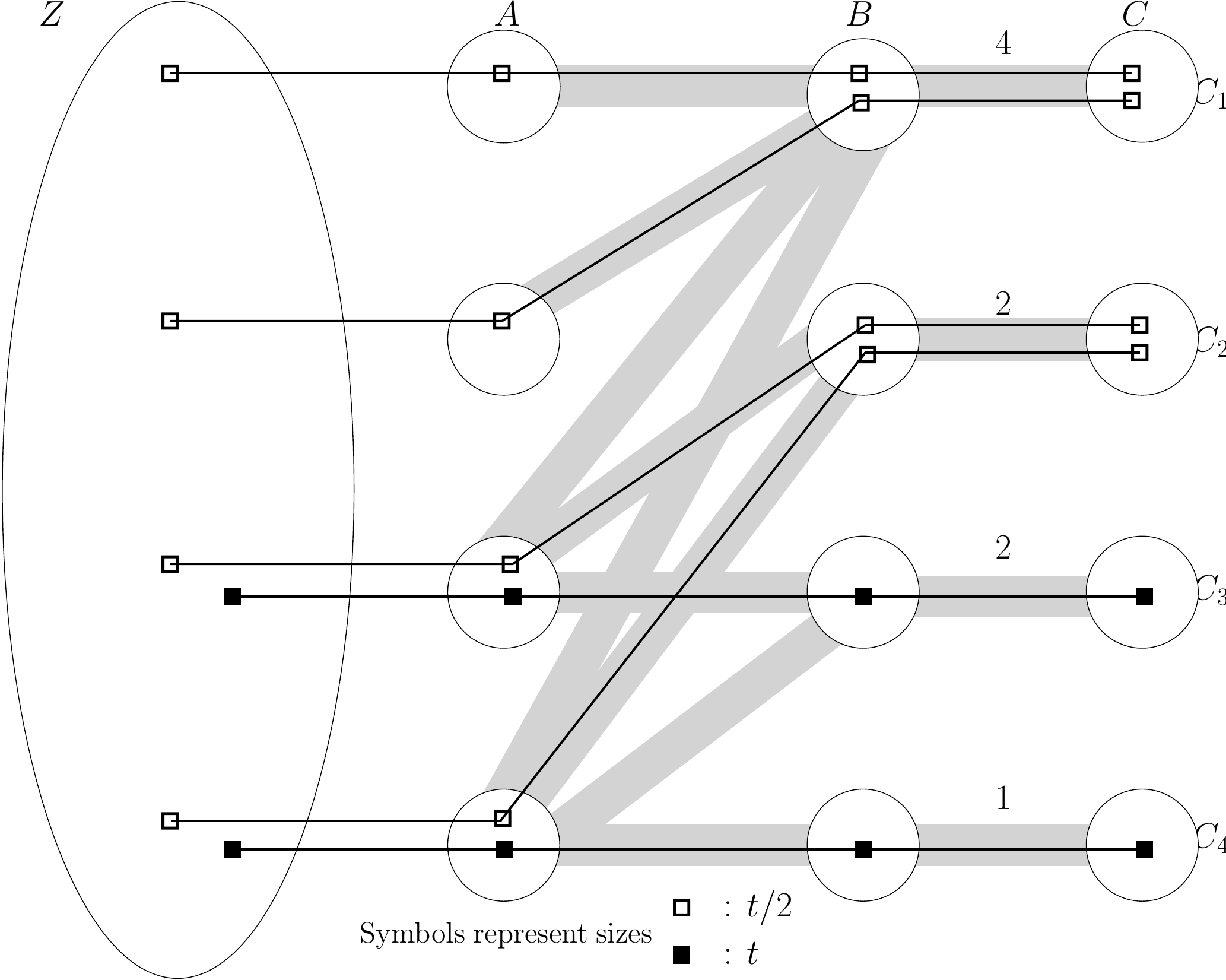} 
\caption{\footnotesize{Extending $(A,B,C)$ when $L_{abc}$ has an $H_{4221}$: The shaded boxes represent a triplet {\em connected} to $Z$, while solid thick lines represent a balanced complete $4$ partite $4$-graphs}}
\label{fig_4_unif_H_4221_extended}
\end{figure}
Again we repeat this process so that we remove $t \geq \eta m/24$ vertices from each color class of $B$ and $C$. Note that with this process we have used $3t/2$ vertices in $A_3$ and $A_4$ while $t/2$ vertices each in $A_1$ and $A_2$ (see Figure \ref{fig_4_unif_H_4221_extended}). Furthermore only the remaining part of $A$ is not balanced. The balance can be restored by discarding $t$ vertices each from the remaining part of $A_1$ and $A_2$ which results in the net increase of $2t\geq \eta m/12$ vertices in all the balanced complete tripartite graphs.
\vskip3pt
\noindent {\bf Case 4:} $L_{abc}$ has a subgraph isomorphic to $H_{3321}$: 

\vskip2pt \noindent
Assume that the pairs corresponding to the degree $3$, $3$, $2$ and $1$ pairs in the $H_{3321}$ are $(B_1,C_1)$, $(B_2,C_2)$, $(B_3,C_3)$ and $(B_4,C_4)$ respectively. Let the color classes corresponding to the neighbors of these pairs in $H_{3321}$ be $\{A_2,A_3,A_4\}$, $\{A_2,A_3,A_4\}$, $\{A_3,A_4\}$ and $\{A_4\}$. 
By the definition of connectedness and Lemma \ref{4partVolArg} we build three complete balanced $4$-partite graphs in each of $(A_4,B_4,C_4,Z)$, $(A_3,B_3,C_3,Z)$ and $(A_2,B_2,C_2,Z)$ of size $\beta\sqrt{\log m}$. In addition we build three more complete balanced $4$-partite $4$-graphs using $(B_1,C_1)$ and $Z$, while the vertices of the fourth color classes are in $A_2$, $A_3$ and $A_4$ respectively.  We repeat this process so as to remove at least $t \geq 3\eta m/64$ vertices each from each color class of $B$ and $C$. Clearly the remaining part of $B$ and $C$ are still balanced, while in $A$ we have used $4t/3$ vertices in each of $A_2,A_3$ and $A_4$. To restore the balance in remaining part of $A$ we discard arbitrary  $4t/3$ vertices from $A_1$. In the process the net increase in the number of vertices in the resultant balanced $4$-partite $4$-graphs is at least $8t/3 \geq 3\eta m/8$, hence $(A,B,C)$ is {\em extended}.\hfill{} \end{proof}
\begin{figure}[h!] 
\centering
\includegraphics[scale=0.575]{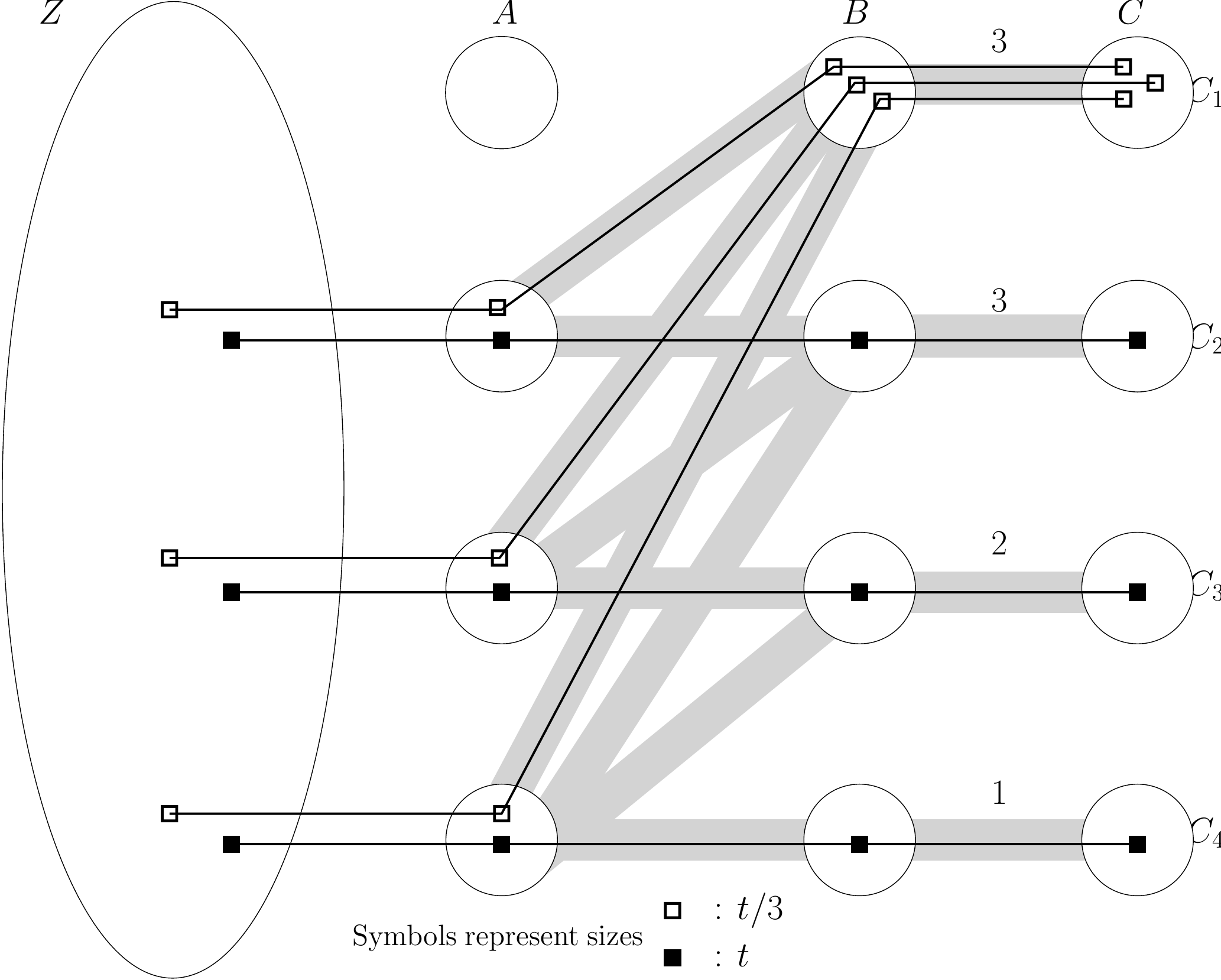} 
\caption{\footnotesize{Extending $(A,B,C)$ when $L_{abc}$ has an $H_{3321}$: The shaded boxes represent a triplet {\em connected} to $Z$, while solid thick lines represent a balanced complete $4$ partite $4$-graphs}}
\label{fig_4_unif_H_3321_extended}
\end{figure}


\section{The Non Extremal Case}\label{non_ext_case}

Throughout this section we assume that we have a $4$-graph $H$ satisfying (\ref{minDegree}) such that the extremal case does not hold for $H$. We shall assume that $n$ is sufficiently large and besides our main parameter $\gamma$ we use the parameters $\beta$ and $\alpha$ such that the following holds  \beq\label{para} \frac{12}{2^{1/\beta^2}} < \gamma = \alpha^4 \ll 1 \eeq where $a\ll b$ means that $a$ is sufficiently small compared to $b$. From (\ref{minDegree}) and (\ref{para}), when $n$ is large we have $$\delta_1(H)\geq{n-1 \choose 3}-{3n/4\choose 3} +1 > \frac{37}{64}{n-1\choose 3} - \frac{9n^2}{64} > \left(1/2 + 2\sqrt{\alpha}\right){n\choose 3}$$ Hence our hypergraph $H$ satisfies the conditions of Lemma \ref{absorbLemma} (the absorbing lemma). We remove from $H$  an {\em absorbing matching} $M$ of size at most $\alpha^2 n/4 = \sqrt{\gamma}n/4$. In the remaining hypergraph we find an almost perfect matching that leaves out a set of at most $\alpha^4n = \gamma n$ vertices. As guaranteed by Lemma \ref{absorbLemma} the vertices that are left out from the almost perfect matching can be absorbed into $M$, therefore we get a perfect matching in $H$. In what follows we work with the remaining hypergraph (after removing $V(M)$). For simplicity we still denote the remaining hypergraph by $H$ and assume that it is on $n$ vertices. Since $|V(M)|\leq \sqrt{\gamma}n$, in the remaining hypergraph we still have \beq\label{minDegOrig}\delta_1(H)\geq \left(\frac{37}{64}-6\sqrt{\gamma}\right){n \choose 3}\eeq as for any vertex $v$ in the remaining hypergraph, there can be at most $6\sqrt{\gamma}{n\choose 3}$ edges containing $v$ and at least one vertex in $V(M)$.

\subsection{The optimal cover}
Our goal is to find an almost perfect matching in $H$. We are going to build a cover ${\cal F} = \{Q_1,Q_2,\ldots$\}, such that, each $Q_i$ is a disjoint balanced complete $4$-partite $4$-graph in $H$ (we refer to them as $4$-partite graphs). We say that such a cover is {\em optimal} if it covers at least $(1-\gamma)n$ vertices. We will show that either we can find an optimal cover or $H$ is $\alpha$-extremal. It is easy to see that such an optimal cover readily gives us an almost perfect matching.
\vskip6pt
Using the following iterative procedure, we either build an optimal cover or find a subset of vertices, which shows that $H$ is $\alpha$-extremal. We begin with a cover ${\cal F}_0$. Then in each step $t\geq 1$, if ${\cal F}_{t-1}$ is not optimal, we find another cover ${\cal F}_{t}$, such that $|V({\cal F}_{t})|\geq |V({\cal F}_{t-1})| +\gamma^2 n/16$ (for this we use the notation, ${\cal F}_{t} > {\cal F}_{t-1}$). The size of a color class in each $4$-partite graph in ${\cal F}_t$ is $m_t$.  
\vskip6pt
To get the initial cover, ${\cal F}_{0}$, we repeatedly apply Lemma \ref{hyperKST} in the leftover of $H$, while the conditions of the lemma are satisfied and the number of leftover vertices are at least $\gamma n$, to find $4$-partite graphs, $K^{(4)}(m_0)$, where $m_0 = \beta(\log n)^{1/3}$. After the $t^{\mbox{\emph{th}}}$ step in this iterative procedure, if ${\cal F}_{t}$ is not an optimal cover, then we get ${\cal F}_{t+1}$. We will show that, unless $H$ is $\alpha$-extremal, we have $ {\cal F}_{t+1} > {\cal F}_{t}$ and $m_{t+1} = \beta\sqrt{\log m_{t}}$. Let ${\cal I}_t = V(H)\setminus V({\cal F}_t)$. Since ${\cal F}_t$ is not optimal and we cannot apply Lemma \ref{hyperKST} in $H|_{{\cal I}_t}$, we must have that $|{\cal I}_t| \geq \gamma n$ and \beq \label{denI} d_4({\cal I}_t) < \gamma .\eeq By non-extremality of $H$, this implies that $|V({\cal F}_0)| \geq n/4$.


\vskip4pt

\noindent In what follows we will show that if there are `many' edges with three vertices in ${\cal I}_t$ and one vertex in some $Q_i\in {\cal F}_t$ then we get ${\cal F}_{t+1} > {\cal F}_{t}$. To that end let $Q_i = (V_1^i,V_2^i,V_3^i,V_4^i)$ be a $4$-partite graph in ${\cal F}_t$, we say that a color class, $V_p^i$ of $Q_i$, $(1\leq p\leq 4)$ is {\em connected} to ${\cal I}_t$, if $d_4(V_p^i,{{\cal I}_t\choose 3}) \geq 2\gamma$. 
We will show that if a $4\gamma$-fraction of the $4$-partite graphs in ${\cal F}_t$ have at least $2$ color classes connected to ${\cal I}_t$, then we can we get ${\cal F}_{t+1} > {\cal F}_{t}$. To see this assume that we have a subcover ${\cal F}' \subset {\cal F}_t$, such that each $Q_i\in  {\cal F}'$ has at least $2$ color classes connected to ${\cal I}_t$ and $|V({\cal F}')|\geq \gamma n$. Let ${\cal F}' = \{Q_1,Q_2,\ldots\} \subset {\cal F}_t $ and without loss of generality, say $V_1^i$ and $V_2^i$ are the color classes in each $Q_i$ that are connected to ${\cal I}_t$. For each such $Q_i$, since $|V_1^i|=|V_2^i| =m_t \leq \beta(\log n)^{1/3}$ and $|{\cal I}_t|\geq \gamma n$, by Lemma \ref{subsetPHP_hyper} we can find two disjoint balanced complete $4$-partite graphs $(U_1^i,A_1^i,B_1^i,C_1^i)$ and $(U_2^i,A_2^i,B_2^i,C_2^i)$ where $U_1^i$ and $U_2^i$ are subsets of $V_1^i$ and $V_2^i$ respectively, and $A_k^i, B_k^i \mbox{ and } C_k^i$, $k\in\{1,2\}$ are disjoint subsets of ${\cal I}_t$. The size of each color class of these new $4$-partite graphs is at least $\gamma m_t/4$ (see Figure \ref{4_unif_2sided_extension}). 

\begin{figure}[h!] 
\centering
\includegraphics[scale=0.65]{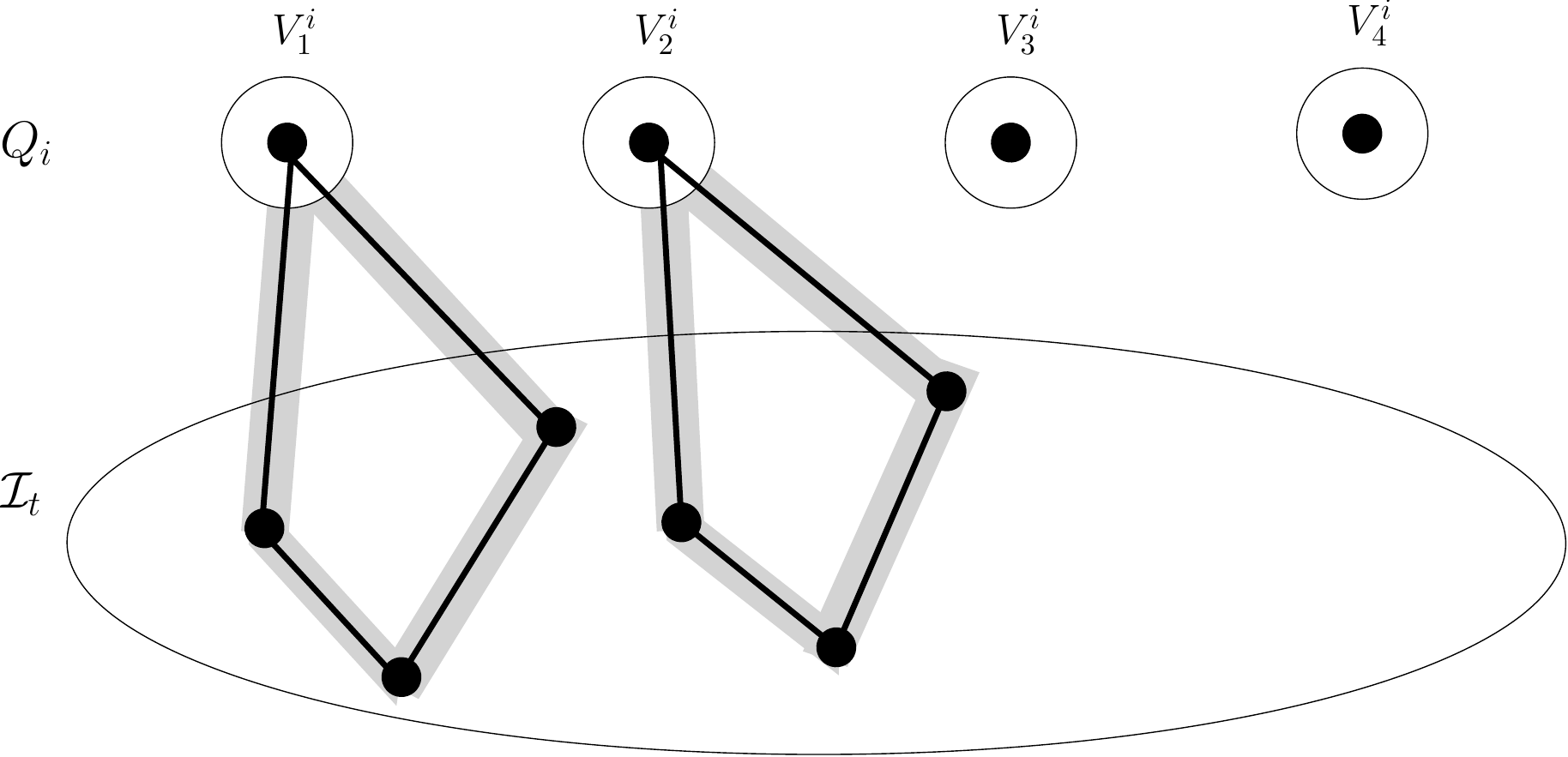} 
\caption{\footnotesize{$V_1^i$ and $V_2^i$ of $Q_i$ are connected to ${\cal I}_t$. The Shaded quadrilaterals represent density, while the dark black quadrilateral is balanced complete $4$-partite $4$-graph}}
\label{4_unif_2sided_extension}
\end{figure}

We remove the vertices of these new $4$-partite graphs from their respective sets and add them to ${\cal F}_{t+1}$. Removing these vertices from $V_1^i$ and $V_2^i$ creates an imbalance in the leftover part of $Q_i$ ($V_3^i$ and $V_4^i$ have more vertices). To restore the balance, we discard (add to ${\cal I}_t$) some arbitrary $|U_1^i| = |U_2^i|$ vertices each from $V_3^i$ and $V_4^i$. The new $4$-partite graphs use $6\gamma m_t/4$ vertices from ${\cal I}_t$. Therefore even after discarding the vertices from $V_3^i$ and $V_4^i$ the net increase in the size of our cover is $\gamma m_t$, while all the $4$-partite graphs are balanced. 
We repeat this procedure for every $Q_i \in {\cal F}'$ and add the leftover part of $Q_i$ and all $4$-partite graphs in ${\cal F}_t \setminus {\cal F}'$, to ${\cal F}_{t+1}$. Now, we split each $4$-partite graph in ${\cal F}_{t+1}$ into disjoint balanced complete $4$-partite graphs, such that each has a color class of size $m_{t+1} = \beta\sqrt{\log m_t}$ (we assume divisibility). Since $|V({\cal F}')|\geq \gamma n$, by the above observation, we have $|V({\cal F}_{t+1})| \geq |V({\cal F}_t)| + \gamma^2 n$, hence ${\cal F}_{t+1} > {\cal F}_{t}$.  
\vskip6pt

\vskip6pt
Similarly, if there are `many' edges that uses two vertices in ${\cal I}_t$ and two vertices from $V({\cal F}_t)$ then we get ${\cal F}_{t+1} > {\cal F}_{t}$.
First note that the number of pairs of vertices of any $4$-partite graph $Q_i \in {\cal F}_t$ is $O(\log n)^{2/3}$, so the number of pairs of vertices within the $4$-partite graphs in ${\cal F}_t$ is $O(n(\log n)^{2/3}) = o{n\choose 2}$. Therefore the total number of edges, containing two vertices within a $Q_i \in {\cal F}_t$ and two vertices in ${\cal I}_t$ is $o{n\choose 4}$, hence we ignore such edges. Let $Q_i = (V_1^i,V_2^i,V_3^i,V_4^i)$ and $Q_j = (V_1^j,V_2^j,V_3^j,V_4^j)$ be a pair of $4$-partite graphs in ${\cal F}_t$ we say that ${\cal I}_t$ is {\em connected} to a pair of color classes $(V_p^i,V_q^j)$, $(1\leq p,q\leq 4)$, if $d_4\left(V_p^i, V_q^j,{{\cal I}_t\choose 2}\right)\geq 2\gamma$. We say that ${\cal I}_t$ is {\em $k$-sided} to a pair $(Q_i,Q_j) \in {{\cal F}_t\choose 2}$ if ${\cal I}_t$ is connected to $k$-pairs in $\{V_1^i,V_2^i,V_3^i,V_4^i\}\times\{V_1^j,V_2^j,V_3^j,V_4^j\}$. Assume that ${\cal I}_t$ is at least $9$-sided to a $4\gamma$-fraction of pairs of $4$-partite graphs in ${{\cal F}_t\choose 2}$. By a simple greedy procedure, (Lemma \ref{folklore_min_degree_subgraph} and the $2$-graph analog of Lemma \ref{folklore_matching}) we get a disjoint set of pairs, $M' \subset {{\cal F}_t\choose 2}$, such that for each pair $(Q_i,Q_j)\in M'$, ${\cal I}_t$ is at least $9$-sided to $(Q_i,Q_j)$ and the number of vertices covered by $M'$ is at least $\gamma n$.

\begin{figure}[h!] 
\centering
\includegraphics[scale=0.65]{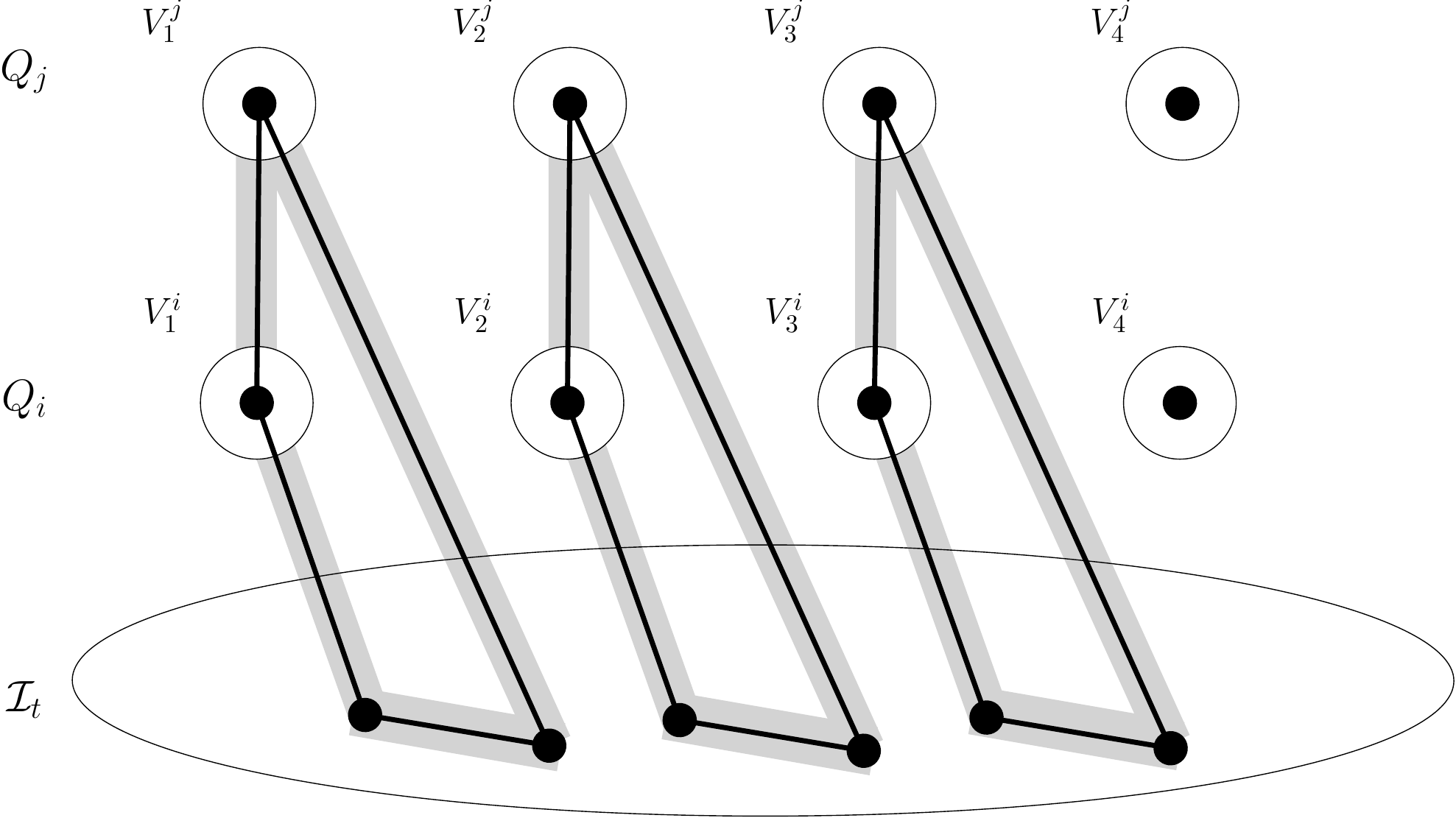} 
\caption{\footnotesize{The pair ($Q_i,Q_j$) is at least $9$-sided to ${\cal I}_t$. The Shaded quadrilaterals represent density, while the dark black quadrilateral is balanced complete $4$-partite $4$-graph}}
\label{4_unif_9Sided_pair_extended}
\end{figure}

\noindent For every $(Q_i,Q_j) \in M'$ we proceed as follows: Since ${\cal I}_t$ is connected to at least $9$ pairs of color classes in $\{V_1^i,V_2^i,V_3^i,V_4^i\}\times\{V_1^j,V_2^j,V_3^j,V_4^j\}$, it is easy to see that we can find $3$ disjoint pairs of color classes such that ${\cal I}_t$ is connected to each of them (say $(V_1^i,V_1^j),(V_2^i,V_2^j) \mbox{ and } (V_3^i,V_3^j)$ are such disjoint pairs of color classes). We have $|V_1^i|=|V_1^j| = m_t \leq \beta(\log n)^{1/3}$ and $|{\cal I}_t|\geq \gamma n$, so by definition of connectedness, the induced hypergraph, $H\left(V_1^i,V_1^j,{{\cal I}_t\choose 2}\right)$ satisfies the conditions of Lemma \ref{subsetPHP_hyper_2}. Therefore by Lemma \ref{subsetPHP_hyper_2} we remove balanced complete $4$-partite $4$-graphs, such that each of them has one color class in $V_1^i$, one in $V_1^j$ and two color classes in ${\cal I}_t$. Note that the conditions of Lemma \ref{subsetPHP_hyper_2} are satisfied until in total we remove at least $\gamma m_t/2$ vertices each from  $V_1^i$ and $V_1^j$. We repeat the same process with $(V_2^i,V_2^j) \mbox{ and } (V_3^i,V_3^j)$ (see Figure \ref{4_unif_9Sided_pair_extended}). In total these new balanced complete $4$-partite $4$-graphs use $3\gamma m_t$ vertices from ${\cal I}_t$ and when we remove the vertices of the new $4$-partite graphs, from $Q_i$ and $Q_j$ the remaining parts of $Q_i$ and $Q_j$ are not balanced ($V_4^i$ and $V_4^j$ have extra vertices). To restore the balance we remove an arbitrary set of vertices each from $V_4^i$ and $V_4^j$ (equal to the difference in color classes). Still the net increase in the number of vertices in the cover is at least $2\gamma m_t$. 

\vskip 6pt

Applying this procedure for every $(Q_i,Q_j) \in M'$ and making all of the $4$-partite graphs of the same size, by splitting, we get ${\cal F}_{t+1}$. Again we have all $4$-partite graphs balanced with color classes of size $m_{t+1} = \beta\sqrt{\log m_t}$ and ${\cal F}_{t+1} > {\cal F}_{t}$.
\vskip 10pt

\noindent Now, if we can not get ${\cal F}_{t+1} > {\cal F}_{t}$, we must have that, in almost all $4$-partite graphs in ${\cal F}_{t}$, at most one color class is connected to ${\cal I}_t$ and almost all pairs of $4$-partite graphs in ${{\cal F}_{t}\choose 2}$ are at most $8$-sided. In particular this implies that for a typical vertex $v\in {\cal I}_t$ we have $$deg_4\left(v,{{\cal I}_t\choose 2}\times V\left({\cal F}_t\right)\right)\leq \left(\frac{1}{4} + 6\gamma\right)|V\left({\cal F}_t\right)|{|{\cal I}_t|\choose 2}$$ and $$deg_4\left(v, {\cal I}_t\times {V\left({\cal F}_t\right)\choose 2}\right) \leq \left(\frac{1}{2} + 16\gamma\right) |{\cal I}_t|{|V\left({\cal F}_t\right)|\choose 2}.$$

\noindent From (\ref{minDegOrig}), (\ref{denI}) and the above degree bounds, for a typical vertex $v\in{\cal I}_t$ we have 
\begin{align}
&deg_4\left(v,{V({\cal F}_t)\choose 3}\right)\notag{}\\
 \geq& \left(\frac{37}{64}-6\sqrt{\gamma}\right){n \choose 3} - d_4\left(v, {\cal I}_t \times {V({\cal F}_t)\choose 2}\right) - d_4\left(v, {{\cal I}_t\choose 2} \times V({\cal F}_t)\right) - d_4\left(v,{{\cal I}_t\choose 3}\right)\notag{}\\
\geq & \left(\frac{37}{64}-6\sqrt{\gamma}\right){n \choose 3} - \left(\frac{1}{2} + 16\gamma\right) |{\cal I}_t|{|V({\cal F}_t)|\choose 2} - \left(\frac{1}{4}+6\gamma\right) {|{\cal I}_t|\choose 2}|V({\cal F}_t)| - \gamma {|{\cal I}_t|\choose 3}\notag{}\\
\geq & \left(\frac{37}{64}-30\sqrt{\gamma}\right){|V({\cal F}_t)| \choose 3}\notag{}
\end{align}
where the last inequality holds when $|{\cal I}_t| \geq \gamma n$ and $|V({\cal F}_t)|\geq n/4$. 

\vskip6pt
For a vertex $v$, consider the edges that $v$ makes with $3$-sets of vertices within a $Q_i \in {\cal F}_t$. The number of $3$-sets of vertices of any $Q_i \in \cal{F}_t$ is $O(\log n)$, (as the size of $Q_i$ is at most $\beta\log^{1/3}n$), hence the total number of $3$-sets of vertices within the $4$-partite graphs in ${\cal F}_t$ is $O(n\log n) = o{n\choose 3}$. Similarly the number of $3$-sets of vertices which uses $2$ vertices from a $Q_i$ and one vertex from some other $Q_j$ is $O(n^2\log n) = o{n\choose 3}$. Therefore, for any vertex $v$, we ignore these types of edges and we will only consider the edges that $v$ makes with $3$-sets of vertices $(x,y,z),\;x \in V(Q_i),\;y \in V(Q_j),\; z\in V(Q_k),\; i\neq j\neq k$. By the above observation, for the minimum degree of a typical vertex $v\in {\cal I}_t$ we still have \beq \label{min_cross_Degree} deg_4\left(v,{V({\cal F}_t)\choose 3}\right)\geq \left(\frac{37}{64}-40\sqrt{\gamma}\right){|V({\cal F}_t)| \choose 3} \eeq 

\vskip 6pt

\noindent Let $Q_i = (V_1^i,\ldots,V_4^i)$, $Q_j = (V_1^j,\ldots,V_4^j)$ and $Q_k = (V_1^k,\ldots,V_4^k)$ be three $4$-partite graphs in ${\cal F}_t$, we say that ${\cal I}_t$ is {\em connected} to a triplet of color classes $(V_p^i,V_q^j,V_r^k)$, $1\leq p,q,r \leq 4$, if $d_4({\cal I}_t,(V_p^i \times V_q^j\times V_r^k)) \geq 2\gamma$. For $(Q_i,Q_j,Q_k) \in {{\cal F}_t \choose 3}$ we consider the link graph $L_{ijk}$ as defined above (with ${\cal I}_t$ playing the role of the set $Z$).
\vskip6pt

For a constant $\eta >0$, we say that ${\cal I}_t$ is $(\eta, s)$-{\em connected} to ${\cal F}_t$, if there is a subset of triplets of $4$-partite graphs, $T \subset  {{\cal F}_t \choose 3}$, such that for each triplet $(Q_i,Q_j,Q_k) \in T$, the link graph, $L_{ijk}$ has $s$ edges and $|T| \geq \eta {|{\cal F}_t| \choose 3}$. 

\vskip6pt
A simple calculation, using (\ref{min_cross_Degree}), implies that if ${\cal I}_t$ is $(\gamma^{1/3}, s)${\em-connected} to ${\cal F}_t$ for some $s\leq 36$, then we also have that ${\cal I}_t$ is $(\sqrt{\gamma}, \geq38)${\em-connected} to ${\cal F}_t$.


%
 

\vskip4pt

\vskip 10pt

We consider the following cases based on the way ${\cal I}_t$ is connected to ${\cal F}_t$ and show that either we get ${\cal F}_{t+1} > {\cal F}_{t}$ or $H$ is extremal. First assume that ${\cal I}_t$ is $(32\gamma, \geq37)$-{\em connected} to ${\cal F}_t$ such that for every triplet $(Q_i,Q_j,Q_k) \in T$, the link graph $L_{ijk}$ is not isomorphic to $H_{ext}$. Then by lemma \ref{folklore_min_degree_subgraph} and lemma \ref{folklore_matching} there exists a set a disjoint set of triplets of $4$-partite graphs, $T' \subset T$, such that for each triplet $(Q_i,Q_j,Q_k) \in T'$ the link graph $L_{ijk}$ has at least $37$ edges and is not isomorphic to $H_{ext}$. Furthermore, the number of vertices covered by $T'$ is at least $\gamma n$.

\vskip 4pt
\noindent Now, for each $(Q_i,Q_j,Q_k) \in {T}'$, since $L_{ijk}$ has at least $37$ edges and $L_{ijk}\neq H_{ext}$, using Lemma \ref{extensionLemma} we extend $(Q_i,Q_j,Q_k)$ to add at least $\gamma m_t/16$ vertices to our cover. Clearly if we extend every triplet in $T'$ the net increase in the size of our cover is at least ${\gamma}^2n/16$ (as the size of $T'$ is at least $\gamma n$). Similarly as above we can split the $4$-partite graphs to make them of the same size and get ${\cal F}_{t+1} > {\cal F}_{t}$.

\vskip 10pt 
\noindent On the other hand, if there is no such $T$, then we  must have that ${\cal I}_t$ is not $(\gamma^{1/3},s)$-{\em connected} to ${\cal F}_t$ for any $s\leq 36$, because otherwise, as observed above, we will get such a $T$. So, roughly speaking, for almost all triplets of $4$-partite graphs in ${{\cal F}_t\choose 3}$, we have that the link graph of the triplet has exactly $37$ edges and is isomorphic to $H_{ext}$. Call a $4$-partite graph $Q_i \in {\cal F}_t$ {\em good}, if for almost all pairs of other $4$-partite graphs $Q_j,Q_k$, we have that the link graph $L_{ijk}$ is isomorphic to $H_{ext}$. 
\vskip 10pt
By the above observation, almost all $4$-partite graphs (covering $\geq (1-2\gamma^{1/3})|V({\cal F}_t)|$ vertices) are good. By a simple greedy procedure we find a set of disjoint triplets of $4$-partite graphs, $T_g$, such that for each triplet $(Q_i,Q_j,Q_k) \in T_g$, the link graph $L_{ijk}$ is isomorphic to $H_{ext}$ and all good $4$-partite graphs are part of some triplet in $T_g$. Let the set of $4$-partite graphs covered by $T_g$ be ${\cal F}_g$, clearly $|V({\cal F}_g)|\geq (1-2\gamma^{1/3})|V({\cal F}_t)|$. With relabeling we may also assume that in each triplet $(Q_i,Q_j,Q_k) \in T_g$, $V_1^i,V_1^j$ and $V_1^k$ are the color classes corresponding to the vertices of the link graph $L_{ijk}$ that intersect every edge of $L_{ijk}$.  
\vskip 10pt 

For every $(Q_i,Q_j,Q_k) \in T_g$, by definition of connectedness and the sizes of the $Q_i,Q_j$ and $Q_k$, the $4$-partite hypergraph induced by $(V_1^i,V_2^j,V_2^k,{\cal I}_t)$ satisfies the conditions of Lemma \ref{4partVolArg}. Hence applying Lemma \ref{4partVolArg} we find a balanced complete $4$-partite graphs $X_{i1}$ in $H(V_1^i,V_2^j,V_2^k,{\cal I}_t)$. We also find two more disjoint balanced complete $4$-partite graphs, $X_{i2}$ and $X_{i3}$ in $H(V_1^i,V_3^j,V_3^k,{\cal I}_t)$ and $H(V_1^i,V_4^j,V_4^k,{\cal I}_t)$. Since $|{\cal I}_t|\geq \gamma n$, we can find these complete $4$-partite graphs that are disjoint from each other. The sizes of a color class in each of $X_{11},X_{12}$ and $X_{13}$ is $\beta\sqrt{\log m_t}$.  
\vskip 10pt
 
Similarly for each of $V_1^j$ and $V_1^k$ we find $3$ disjoint balanced complete $4$-partite graphs $X_{jp}$ and $X_{kp}$ in $H(V_p^i,V_1^j,V_p^k,{\cal I}_t)$ and $H(V_p^i,V_p^j,V_1^k,{\cal I}_t)$, $(2\leq p\leq 4)$, respectively. All of these these balanced complete $4$-partite graphs are disjoint from each other and the size of a color class in each one of them is $\beta\sqrt{\log m_t}$. 

By the definition of connectedness, the structure of the link graph $L_{ijk}$ and the fact that $|{\cal I}_t|\geq \gamma n$, clearly we can find these nine disjoint $4$-partite graphs. And as argued above we repeat this process (remove another set of $9$ such $4$-partite graphs) until in total we remove $\gamma m_t/2$ vertices from each of $V_1^i$, $V_1^j$ and $V_1^k$, while $\gamma m_t/3$ vertices from each of the other classes in $Q_i,Q_j$ and $Q_k$. 
\vskip 10pt

Note that these tripartite graphs in total use $3\gamma m_t/2$ vertices from ${\cal I}_t$. But this creates an imbalance among the color classes of the remaining parts of $Q_i$, $Q_j$ and $Q_k$ ($V_1^i,V_1^j$ and $V_1^k$ have fewer vertices), to restore the balance we will have to discard $\gamma m_t/6$ vertices from each color class in $Q_i$, $Q_j$ and $Q_k$ except $V_1^i$, $V_1^j$ and $V_1^k$. Which leaves us with no net gain in the size of the cover. Therefore we will not discard any vertices from these color classes at this time and say that these color classes have extra vertices. We proceed in similar manner for each triplet in $T_g$. So we have about $\gamma n/24$ extra vertices altogether.\\

\noindent Denote by $V_1^g,V_2^g,V_3^g$ and $V_4^g$ the union of the corresponding color classes of remaining parts of $4$-partite graphs in ${\cal F}_g$. Clearly $|V_2^g| = |V_3^g| = |V_4^g| \geq (1-2\gamma^{1/3})|V({\cal F}_t)|/4 - \gamma n /24 \geq (1/16 - 3\gamma^{1/3})n$. The last inequality follows from the lower bound on size of $V({\cal F}_t)$ above, and the fact that $\gamma$ is a small constant. We will show that either we can increase the size of our cover or we have \beq\label{emptyV2V3V4} d_4\left(V_2^g\cup V_3^g\cup V_4^g\right) \leq \sqrt{\gamma}.\eeq 


\noindent For $d_4(V_2^g\cup V_3^g\cup V_4^g)$, we only consider those edges that use exactly one vertex from a $4$-partite graph $Q_i$, as the number of edges of other types is $o(n^4)$. Assume that $d_4(V_2^g\cup V_3^g\cup V_4^g) \geq \sqrt{\gamma}$ then by Lemma \ref{hyperKST} there exist complete $4$-partite graphs in $H|_{V_2^g\cup V_3^g\cup V_4^g}$ covering at least $\gamma n$ vertices. We remove some of these $4$-partite graphs (possibly with splitting and discarding part of it) such that from no color class we remove more then the number of extra vertices in that color class. Adding these new $4$-partite graphs to our cover increases the size of our cover by at least $\gamma^2 n$ vertices. As we will not need to discard vertices from $(V_2^g\cup V_3^g\cup V_4^g)$ for rebalancing. Instead the extra vertices are part of these new $4$-partite graphs. In the remaining parts of $V_2^g$, $V_3^g$ and $V_4^g$ we arbitrarily discard some extra vertices to restore the balance in the $4$-partite graphs.\\ 

\noindent Similarly we will show that either we can increase the size of our cover or we have \beq\label{emptyV2V3V4_to_I} d_4\left(V_2^g\cup V_3^g \cup V_4^g,{{\cal I}_t\choose 3}\right) \leq \sqrt{\gamma} .\eeq Indeed assume the contrary, i.e. $d_4(V_2^g\cup V_3^g\cup V_4^g,{{\cal I}_t\choose 3}) \geq \sqrt{\gamma}$, then since both $|{\cal I}_t|$ and $|V_2\cup V_3^g\cup V_4^g|$ are at least $\gamma n$, by Lemma \ref{hyperKST} we can find disjoint complete $4$-partite graphs with one color class in $V_2\cup V_3^g\cup V_4^g$ and three in ${\cal I}_t$ covering at least $\gamma^2n/2$ vertices. And again as above we can add these $4$-partite graphs to our cover and increase the size of our cover as we have extra vertices in $V_2^g\cup V_3^g\cup V_4^g$. By the same reasoning we can prove that there are very few edges that uses two vertices from $V_2\cup V_3^g\cup V_4^g$ and two from ${\cal I}_t$.\\ 

From the above observations about the size of ${\cal F}_g$ and (\ref{para}) we have that $|V_2^g\cup V_3^g \cup V_4^g \cup {\cal I}_t| \geq (3/4 - \alpha)n$. Therefore if we can not increase the size of our cover significantly (by at least $\gamma^2n/16$ vertices), then by (\ref{denI}), (\ref{emptyV2V3V4}) and (\ref{emptyV2V3V4_to_I}) we get that $d_4(V_2^g\cup V_3^g\cup V_4^g \cup {\cal I}_t) < \alpha$. Hence $H$ is $\alpha$-extremal. 

\vskip 25pt

\section{The Extremal Case}\label{extCase}

Here our graph $H$ is in the Extremal Case, i.e. {\em there is a $B\subset V(H)$ such that} 
\begin{itemize}
\item $|B|\geq (\frac{3}{4}-\alpha) n$
\item $d_4(B) < \alpha$.
\end{itemize}

\noindent We assume that $n$ is sufficiently large and $\alpha$ is a sufficiently small constant $< 1$. Let $A=V(H)\setminus B$, by shifting some vertices between $A$ and $B$ we can have that $A=n/4$ and $B=3n/4$ as $n \in 4{\mathbb{Z}}$ (we still keep the notation $A$ and $B$). It is easy to see that we still have \beq\label{extDen}d_4(B) < 6\alpha\eeq  

\noindent Since we have $$\delta_1(H) \geq {n-1\choose 3} - {3n/4\choose 3} + 1={n-1\choose 3} - {|B|\choose 3} + 1$$this together with (\ref{extDen}) implies that almost all $4$-sets of $V(H)$ are edges of $H$ except $4$-sets of $B$. Thus roughly speaking we have that almost every vertex $b\in B$ makes edges with almost all $3$-sets of vertices in ${A\choose 3}$,  with almost all $3$-sets of vertices in $B\setminus\{b\} \times {A\choose 2}$ and with almost all $3$-sets of vertices in ${B\setminus\{b\}\choose 2} \times A$  and vice versa. Therefore, we will basically match every vertex in $A$ with a distinct $3$-set of vertices in ${B\choose 3}$ (disjoint from all $3$-sets matched with other vertices in $A$) to get the perfect matching. However some vertices may be `{\em atypical}', in the sense that they may not have this connectivity structure hence we will first find a small matching that covers all such `{\em atypical}' vertices. For the remaining `{\em typical}' vertices we will show that they satisfy the conditions of K{\"o}nig-Hall theorem, hence we will match every remaining vertex in $A$ with a distinct $3$-sets of remaining vertices in $B$. 

\vskip10pt
A vertex $a\in A$ is called {\em exceptional} if it does not make edges with almost all $3$-sets of vertices in $B$, more precisely if $$deg_4\left(a,{B\choose 3}\right) < \left(1-\sqrt{\alpha}\right){|B|\choose 3}$$A vertex $a\in A$ is called {\em strongly exceptional} if it makes edges with very few $3$-sets in $B$, more precisely if $$deg_4\left(a,{B\choose 3}\right) < {\alpha}^{1/3}{|B|\choose 3}$$Similarly a vertex $b\in B$ is called {\em exceptional} if it makes edges with many $3$-sets of vertices in $B$ more precisely if $$deg_4\left(b,{B\setminus\{b\}\choose 3}\right) > \sqrt{\alpha}{|B|\choose 3}$$A vertex $b\in B$ is called {\em strongly exceptional} if it makes edges with almost all $3$-sets of vertices in $B$ more precisely if $$deg_4\left(b,{B\setminus\{b\}\choose 3}\right) > (1-\alpha^{1/3}){|B|\choose 3}$$
\vskip10pt
\noindent Denote the set of \textit{exceptional} and \textit{strongly exceptional} vertices in $A$ (and $B$) by $X_A$ and $SX_A$ respectively (similarly $X_B$ and $SX_B$). Easy calculations using (\ref{minDegree}) and (\ref{extDen}) yields that $|X_A|\leq 18\sqrt{\alpha}|A|$ and $|X_B|\leq 18\sqrt{\alpha}|B|$ and for the {\em strongly exceptional} sets we have  $|SX_A|\leq 40\alpha|A|$ and $|SX_B|\leq 40\alpha|B|$. The constants are not the best possible but we choose them for ease of calculation.

\vskip10pt
If we have both $SX_B$ and $SX_A$ non empty, (say $b\in SX_B$ and $a\in SX_A$) then since then we can exchange $a$ with $b$ and reduce the size of both $SX_B$ and $SX_A$, as it is easy to see that both $a$ and $b$ are not {\em strongly exceptional} in their new sets. Hence one of the sets $SX_A$ and $SX_B$ must be empty. 

\vskip10pt Assume $SX_B \neq \emptyset$. By definition of $SX_B$, for every vertex $b \in SX_B$, we have $deg_3(b,{B\choose 3}) \geq (1-{\alpha}^{1/3}){|B|\choose 3}$. This together with the bound on the size of $SX_B$ implies that we can greedily find $|SX_B|$ vertex disjoint edges in $H|_B$ each containing exactly one vertex of $SX_B$. We also select $|SX_B|$ other vertex disjoint edges such that each edge has two vertices in $B\setminus X_B$ and the two other vertices are in $A$. We can clearly find such edges because by (\ref{minDegree}) and definition of $X_B$ every vertex in $B\setminus X_B$ makes edges with at least $(1-3\sqrt{\alpha})$-fraction of $3$-sets in $B\times{|A|\choose 2}$. We remove the vertices of these edges from $A$ and $B$ and denote the remaining set by $A'$ and $B'$. Let $|A'|+|B'| = n'$, by the above procedure we have $n' = n-8|SX_B|$, $|A'|=|A|-2|SX_B|$ and $|B'|=|B|-6|SX_B|$ hence we get $|B'| = 3|A'| = 3n'/4$. 
\vskip10pt
In case $SX_A \neq \emptyset$ (and $SX_B = \emptyset$), we will first eliminate the vertices in $SX_A$. Note that in this case any vertex $b\in B$ is exchangeable with any vertex in $SX_A$, because if there is a vertex $b\in B$ such that $deg_4(b,{B\choose 3}) \geq \alpha^{1/3}{|B|\choose 3}$ then we can replace $b$ with any vertex $a\in SX_A$ to reduce the size of $SX_A$ (as the vertex $b$ is not {\em strongly exceptional} in $A$ and $a$ can not be {\em strongly exceptional} in the set $B$). Therefore we consider the whole set $SX_A\cup B$. By (\ref{minDegree}) for any vertex $v\in SX_A\cup B$ we have \begin{align*}deg_4\left(v,{SX_A \cup B\choose 3}\right)\geq& (|SX_A|-1){|B|\choose 2} + {(|SX_A|-1)\choose 2}|B| + {(|SX_A|-1)\choose 3} + 1\\ \geq&  {4(|SX_A|-1)\choose 3}+1\end{align*}
where the last inequality holds when $n$ is large enough and $|SX_A$ is small. So with a simple greedy procedure we find $|SX_A|$ disjoint edges in $H|_{SX_A\cup B}$ and remove these edges from $H$. Note that this is the only place where we critically use the minimum degree. We let $A' = A\setminus SX_A$ and $B'$ has all other remaining vertices. Again as above we have $n' = n-4|SX_A|$, $|A'|=|A|-|SX_A|$ and $|B'|=|B|-3|SX_A|$ hence we get $|B'| = 2|A'| = 3n'/4$. 

\vskip4pt
Having dealt with the \textit{strongly exceptional} vertices, the vertices of $X_A$ and $X_B$ in $A'$ and $B'$ can be eliminated using the fact that their sizes are much smaller than the crossing degrees of vertices in those sets. For instance as observed above we have $|X_A|\leq 18\sqrt{\alpha}|A|$ while for any vertex $a\in X_A$, we have that $deg_4(a,{B'\choose 2}) \geq \alpha^{1/3}{|B'|\choose 3}/2$ (because $a\notin SX_A$). Therefore by a simple greedy procedure for each $a\in X_A$ we delete a disjoint edge that contains $a$ and three vertices from $B'$. Similarly for each $b \in X_B$ we delete an edge that contains $b$ and uses one vertex from $A'$ and the other two vertices from $B'$ distinct from $b$. Clearly we can find such disjoint edges, hence we removed a partial matching that covers all vertices in the {\em strongly exceptional} and {\em exceptional} sets. 

\vskip 10pt
Finally in the leftover sets of $A'$ and $B'$ (denote them by $A''$ and $B''$, by construction we still have $|B''|=3|A''|$) we will find $|A''|$ disjoint edges each using one vertex in $A''$ and three vertices in $B''$. Note that for every vertex  $a\in A''$ we have $deg_4(a,{B''\choose 3})\geq (1-2\alpha^{1/3}){|B''|\choose 3}$ (as $a\notin X_A$). We say that a vertex $b_i$ and a pair $b_j,b_k$ in $B''$ are {\em good} for each other if $(b_i,b_j,b_k,a_l)\in E(H)$ for at least $(1-40{\alpha}^{1/4})|A''|$ vertices $a_l$ in $A''$. We have that any vertex $b_i \in B''$ is {\em good} for at least $(1-40{\alpha}^{1/4}){|B''|\choose 2}$ pairs of vertices in $B''$ (again this is so because $b_i\notin X_B$). We call such a $(b_i,b_j,b_j)$ a good triplet. 
\vskip 10pt
We randomly select a set $T_1$ of $100{\alpha}^{1/4} |B''|$ vertex disjoint $3$-sets of vertices in $B''$. By the above observation with high probability every vertex $a\in A''$ make edges in $H$ with at least $3|T_1|/4$ triplets in $T_1$ and every triplet in $T_1$ makes an edge with at least $3|A''|/4$ vertices in $A''$. In $B''\setminus V(T_1)$ still every vertex is {\em good} for almost all pairs (as the size of $T_1$ is very small). 
\vskip 10pt
We cover vertices in $B''\setminus V(T_1)$ with disjoint good triplets (i.e. the triplet makes an edge in $H$ with at least $(1-40{\alpha}^{1/4})|A''|$ vertices in $A''$. This can be done by considering a $3$-graph with vertex set $B''\setminus V(T_1)$ and all the {\em good} triplets as its edges. As argued above every vertex is good for almost all pairs. We can find a perfect matching in this $3$-graph (see \cite{Khan_3u_vertDeg}). Let the set of triplets in this perfect matching be $T_2$.
\vskip 10pt
Now construct an auxiliary bipartite graph $G(L,R)$, such that $L= A''$ and vertices in $R$ corresponds to the triplets in $T_1$ and $T_2$. A vertex in $a_l\in L$ is connected to a vertex $y\in R$ if the triplet corresponding to $y$ (say $b_i,b_j,b_k$) is such that $(b_i,b_j,b_k,a_l)\in E(H)$. We will show that $G(L,R)$ satisfies the K{\"o}nig-Hall criteria. Considering the sizes of $A''$ and $T_1$ it is easy to see that for every subset $Q\subset R$ if $|Q|\leq (1-40\alpha^{1/4})|A''|$ then $|N(Q)|\geq |Q|$. When $|Q|>(1-40\alpha^{1/4})|A''|$ (using $|B''|=2|A''|$) any such $Q$ must have at least $6|T_1|/10$ vertices corresponding to pairs in $T_1$, hence with high probability $N(Q) = L \geq |Q|$. Therefore there is a perfect matching of $R$ into $L$. This perfect matching in $G(L,R)$ readily gives us a matching in $H$ covering all vertices in $A''$ and $B''$, which together with the edges we already removed (covering {\em strongly exceptional} and {\em exceptional} vertices) is a perfect matching in $H$.\hfill{} \qed

\begin{acknowledgement}
This work is partially supported by a DIMACS grant. We would also like to Endre Szemer\'edi and Abdul Basit for helpful comments.
\end{acknowledgement}

\end{document}